\documentclass[10pt, twosides]{IEEEtran}

\ifCLASSINFOpdf

\else

\fi

\usepackage{mathrsfs}
\usepackage{amsmath}
\usepackage{amsfonts}
\usepackage{amssymb}
\usepackage{amstext}
\usepackage{graphicx}
\usepackage{indentfirst}
\usepackage{array}
\usepackage{cite}
\usepackage{enumerate}
\usepackage{stmaryrd}
\usepackage[linesnumbered,ruled,vlined]{algorithm2e}
\usepackage{multirow}
\usepackage{multicol}
\usepackage{verbatim}
\usepackage{stfloats}
\usepackage{color}
\usepackage{bm}
\usepackage{multirow}
\usepackage{multicol}

\newtheorem{theorem}{\bf{Theorem}}
\newtheorem{lemma}{\bf{Lemma}}

\newtheorem{corollary}{\bf{Corollary}}

\newtheorem{definition}{\bf{Definition}}
\newtheorem{remark}{\bf{Remark}}
\newtheorem{example}{\bf{Example}}
\newtheorem{metric}{\bf{Metric}}

\newcommand{\tabincell}[2]{\begin{tabular}{@{}#1@{}}#2\end{tabular}}
\begin{document}
%
\title{Polar Coded Diversity on Block Fading Channels via Polar Spectrum}

\author{Kai Niu, \IEEEmembership{Member,~IEEE}, Yan Li, \IEEEmembership{Student Member,~IEEE}
\thanks
{
This work is supported by the National Key R\&D Program of China (No. 2018YFE0205501), the National Natural Science Foundation of China (No. 61671080). \protect\\
\indent K. Niu, Y. Li, are with the Key Laboratory of Universal Wireless Communications, Ministry of Education,
Beijing University of Posts and Telecommunications, Beijing, 100876, China (e-mail: \{niukai, Lyan\}@bupt.edu.cn). \protect\\
}}

\maketitle

\begin{abstract}
Due to the advantage of capacity-achieving, polar codes have been extended to the block fading channel whereas most constructions involve complex iterative-calculation. In this paper, we establish a systematic framework to analyze the error performance of polar codes in the case of block mapping and random mapping. For both the mappings, by introducing the new concept, named split polar spectrum, we derive the upper bound on the error probability of polarized channel which explicitly reveals the relationship between the diversity order $L$ and the block-wise weight distribution of the codeword. For the special case $L=2$ in the block mapping, we design the enumeration algorithm to calculate the exact split polar spectrum based on the general MacWilliams identities. For arbitrary diversity order in the random mapping, with the help of uniform interleaving, we derive the approximate split polar spectrum by combining the polar spectrum and the probability of fading pattern for a specific weight. Furthermore, we propose the design criteria to construct polar codes over the block fading channel. The full diversity criterion is the primary target so as to achieve the diversity gain and the product distance criterion requires to maximize the product of the block-wise Hamming distance whereby obtain the coding gain. Guided by these design criteria, the construction metric, named polarized diversity weight (PDW) is proposed to design the polar codes in both mappings. Such a simple metric can construct polar codes with similar or better performance over those based on traditional methods in block fading channel.
\end{abstract}


\IEEEpeerreviewmaketitle

\section{Introduction}
\label{section_I}
\subsection{Relative Research}
\IEEEPARstart{A}{s} a great breakthrough in channel coding theory, polar codes, invented by Ar{\i}kan \cite{Polarcode_Arikan}, provide a constructive method to achieve the capacity of the symmetric channels. The design idea behind the polar codes is the channel polarization. Given a group of $N$ identical binary-input discrete memoryless channels (B-DMC), after the operations of channel combining and splitting, the generated channels demonstrate the polarization phenomenon whereby the capacity of some channels tends to one (good channels) and that of the others tends to zero (bad channels). Thus, the polar coding can be described that the good channels are assigned as the information bits and the bad channels are set to the fixed bits, namely frozen bits. When the polar codes are concatenated by cyclic redundancy check (CRC) code and CRC aided successive cancellation list (SCL) or successive cancellation stack (SCS) decoding (CA-SCL/SCS) are used \cite{CASCL_Niu}, they demonstrate advantages in error performance over the turbo or LDPC codes, especially for the short/medium code length. Due to the merit of outstanding error performance, the CRC-polar code was selected as the coding standard of the control channel in the fifth generation (5G) wireless communication system \cite{5GNR_38212}.

Some efficient construction algorithms of polar codes have been proposed for typical B-DMCs, such as binary erasure channel (BEC), binary symmetric channel (BSC) or binary-input additive white Gaussian noise (BI-AWGN). Originally, Bhattacharyya parameter \cite{Polarcode_Arikan} is recursively calculated to evaluate the reliability of the polarized channel under BEC. Tal and Vardy proposed an upgrading and degrading algorithm \cite{Tal_Vardy} to bound the error probability of the polarized channel under the general B-DMCs. Gaussian approximation (GA) \cite{GA_Trifonov} was used to estimate the error probability in AWGN channel with a low complexity. Generally, the common characteristics of these algorithms are the iterative calculation dependent on some parameters of the original channel (e.g. crossover probability or signal-to-noise ratio (SNR)).

Recently the construction of polar codes in the fast or block fading channels attracted wide attentions. For the fast Rayleigh fading channel, Trifonov \cite{FastFading_Trifonov} first presented an iterative algorithm to calculate and track the diversity order and noise variance of the polarized channels. Based on the same channel model, Liu and Ling \cite{FastFading_Liu} designed the polar codes and polar lattices by using the Tal-Vardy algorithm \cite{Tal_Vardy} to evaluate the upper/lower bound of the Bhattacharyya parameter of the polarized channels. Lately, Zhou \emph{et al.} \cite{FastFading_Zhou} designed a capacity-equivalent BI-AWGN channel of the Rayleigh channel and constructed the polar codes by using the GA algorithm \cite{GA_Trifonov}. All these construction algorithms in fast fading channel belong to the channel-dependent construction and cannot be straightforwardly extended to the polar coding in block fading channel due to the assumption failure of the channel ergodicity.

On the other hand, for the block fading channel, Boutros and Biglieri \cite{BlockFading_Boutros} first investigated the polarization behaviour whereby full diversity and outage probability are evaluated by using the analysis of mutual information. However, the conclusions in \cite{BlockFading_Boutros} are theoretical and only suitable for a block fading channel with two distinct fading values. Almost simultaneously, for the block Rayleigh fading channel with known channel side information (CSI) or channel distribution information (CDI), Bravo-Santos \cite{BlockFading_Bravo-Santos} proposed a recursive calculation of Bhattacharyya parameter to evaluate the reliability of the polarized channel, whereas this method needs a time-consumption Monte-Carlo computation. Subsequently, Si \emph{et al.} \cite{BlockFading_Si} modeled a fading AWGN channel with BPSK modulation as a fading BSC channel and designed a two-stage polar coding over channel uses and fading blocks. They proved that polar codes can achieve the capacity of fading BSC channel. However, the construction of polar codes still depends on the recursive calculation of Bhattacharyya parameter. Recently, Liu \emph{et al.} \cite{BlockFading_Liu} introduced a natural polarization to indicate distinct reliability of the fading block and optimize the permutation pattern between the code bits and the fading blocks. Although such method may obtain performance gain, it depends on the ideal assumption that the fading envelope is available at the encoder. We conclude that the parameter and state of the block fading channel implicitly affect the error probability of the polarized channel. However, these construction algorithms in block fading channel depend on the iterative calculation of some parameters and lack the systematic framework to establish the explicit relationship between the structure of polar codes and the diversity order, which is significant for the practical application.

On the contrary, the explicit and channel-independent construction is more desirable for the implementation of polar coding in the practical system. He \emph{et al.} \cite{PW_He} designed a channel-independent construction, named polarized weight (PW) algorithm, whereby constructing polar codes with almost the same performance as those constructed by GA algorithm. In addition, a fixed construction sequence for all the code configuration is used in 5G standard \cite{5GNR_38212} which is obtained by computer searching \cite{5Gpolar_Bioglio}. Especially, in the previous work \cite{PolarSpectrum_Niu}, we explored the weight distribution of polar codes and introduced a new concept, named polar spectrum to establish a systematic framework to analyze the error probability of polar codes. Furthermore, two explicit metrics were designed to construct polar codes in terms of polar spectrum. In a word, these explicit construction algorithms can generate polar codes with similar or better error performance than those constructed by conventional channel-dependent construction.

\subsection{Motivation}
In this paper, we focus on the polar coding over the block fading channel whereby the fading envelope is unavailable at the encoder. Many works have investigated the theoretical association between the block fading behavior and the coding techniques, such as \cite{CodedDiversity,On_coding,RandomMapping,ReliableRegion} and references therein. The design criteria of coding over the block fading channel should reveal explicitly the relationship between the error performance and the code block diversity, which are significantly different from the standard design criteria for polar coding over AWGN channel or fast fading channel. However, due to the iterative calculation of Bhattacharyya parameter or mutual information, most current construction works over block fading channel implicitly indicate such a relationship. Furthermore, due to the channel-dependent property, these works lack the merit of explicit construction. Therefore, the aim of this paper is to establish a theoretic framework to analyze the behavior of polar codes in block fading channel and derive some analytical and explicit constructions with low complexity.

\subsection{Main Contributions}
In this paper, by introducing a new concept, named split polar spectrum, we establish a systematic framework to design the polar codes in the block fading channel. For two types of typical mappings from the code bits to the fading blocks, namely block mapping and random mapping, we derive the explicit construction metrics. The main contributions of this paper can be summarized as follows.

\begin{enumerate}[1)]
  \item First, the design criteria based on the (split) polar spectrum for the polar coding in block Rayleigh fading channel are established. Two typical mappings are investigated in this paper, that is, block mapping (the codeword is equally partitioned and mapped into the fading blocks) and random mapping (the code bits are firstly interleaved and then mapped into the fading blocks). For both mappings, by analyzing the upper bound on the error probability of the polarized channel, we build two design criteria, that is, full diversity criterion and product distance criterion. As the primary criterion, the former requires the nonzero code bits are distributed over all the fading blocks so as to achieve full diversity order. Furthermore, the latter requires the product of block-wise Hamming distance between the all-zero codeword and a nonzero codeword is maximized so as to obtain the coding advantage. Both the criteria are dominated by the split weight distribution of polar codes, which indicates the number of codeword-partition blocks with the same distribution of block-wise Hamming weight. In the block mapping, the split polar spectrum can be exactly evaluated. Whereas, in the random mapping, the split polar spectrum is approximately estimated by using the polar spectrum with the help of uniform interleaving.

  \item Second, in block mapping, we derive the upper bound on the error probability of the polarized channel composed of the split polar spectrum and the average pairwise error probability (PEP). After averaging and bounding over the distribution of the fading block envelopes, we find that the PEP is dominated by the product distance. On the other hand, the split polar spectrum determines the number of PEP events. The upper bound of the block error rate (BLER) under the SC decoding for block Rayleigh fading channel is also provided. Furthermore, we analyze the (approximate) upper bound on the error probability for the case of low SNR and derive the logarithmic version of the upper bound as the construction metric, named polarized diversity weight (PDW), which is a simple analytical-metric rather than complex iterative-calculation. The enumeration of the split polar spectrum is a time-consumption process for an arbitrary diversity order $L$. Therefore, we only design an efficient enumeration algorithm for the diversity order $L=2$ based on the general MacWilliams identities. 

  \item Finally, in random mapping, we analyze the corresponding upper bound and build the similar construction metric. Thanks to the uniform interleaving, we obtain the approximation of split polar spectrum by combing the polar spectrum and the probability of a fading block pattern for a specific codeword weight. Similarly, we derive the upper bound on the error probability of the polarized channel depending on the approximation of split polar spectrum. Then, we use the logarithmic version of the (approximate) upper bound as the construction metrics (PDW).

      In both mappings, given the pre-calculated (split) polar spectrum, compared with those constructions based on the iterative calculation, these metrics can construct polar codes with a linear complexity. Furthermore, considering the practical application, they can also be transformed to the channel-independent constructions. Simulation results show that these metrics can generate polar codes to achieve the same or better performance as those constructed by the previous methods.

\end{enumerate}

The remainder of the paper is organized as follows. Section \ref{section_II} presents the preliminaries of polar codes, including polar coding, decoding and polar spectrum. For the block mapping, Section \ref{section_III} describes the performance analysis framework of polar codes in block fading channel. By using split polar spectrum, the upper bound on the error probability of the polarized channel and the upper bound on BLER of polar codes for the block Rayleigh fading channel are derived and analyzed. In Section \ref{section_IV}, we design the construction metric and provide the enumeration algorithm of split polar spectrum for the special diversity order $L=2$. On the other hand, for the random mapping, we derive the approximation of split polar spectrum and analyze the corresponding upper bound in Section \ref{section_V}. Then the construction metric is derived and analyzed in Section \ref{section_VI}. Numerical analysis for the upper bounds in term of (approximate) split polar spectrum and simulation results for comparing PDW construction with the traditional construction are presented in Section \ref{section_VII}. Finally, Section \ref{section_VIII} concludes the paper.

\section{Preliminary of Polar Codes and Polar Spectrum}
\label{section_II}
\subsection{Notation Conventions}
In this paper, calligraphy letters, such as $\mathcal{X}$ and $\mathcal{Y}$, are mainly used to denote sets, and the cardinality of $\mathcal{X}$ is defined as $\left|\mathcal{X}\right|$. The Cartesian product of $\mathcal{X}$ and $\mathcal{Y}$ is written as $\mathcal{X}\times \mathcal{Y}$ and $\mathcal{X}^n$ denotes the $n$-th Cartesian power of $\mathcal{X}$. Let $\mathcal{R}$ denote the real-number set. Especially, the hollow symbol, e.g. $\mathbb{C}$, denotes the codeword set of code or subcode. Let $\llbracket a, b \rrbracket$ denote the continuous integer set $\{a,a+1,\cdots,b\}$.

We write $v_1^N$ to denote an $N$-dimensional vector $\left(v_1,v_2,\cdots,v_N\right)$ and $v_i^j$ to denote a subvector $\left(v_i,v_{i+1},\cdots,v_{j-1},v_j\right)$ of $v_1^N$, $1\leq i,j \leq N$. Simultaneously, we use the boldface lowercase letter, e.g. $\mathbf{u}$, to denote a vector. Further, given an index set $\mathcal{A}\subseteq \llbracket 1,N \rrbracket$ and its complement set $\mathcal{A}^c$, we write $v_\mathcal{A}$ and $v_{\mathcal{A}^c}$ to denote two complementary subvectors of $v_1^N$, which consist of $v_i$s with $i\in\mathcal{A}$ or $i\in\mathcal{A}^c$ respectively.

We use $d_H(\mathbf{u},\mathbf{v})$ to denote the Hamming distance between the binary vector $\mathbf{u}$ and $\mathbf{v}$. Given $\forall {\mathbf{a}}, {\mathbf{b}}\in \mathcal{R}^N$, let $\left\|\bf{a}-\bf{b}\right\|$ denote the Euclidian distance between the vector $\bf{a}$ and $\bf{b}$. We use the boldface capital letter, such as $\mathbf{F}_N$, to denote a matrix with dimension $N$. So the notation $\mathbf{F}_N(a:N)$ indicates the submatrix consisting the rows from $a$ to $N$ of the matrix $\mathbf{F}_N$.

Throughout this paper, $\log\left(\cdot\right)$ means ``logarithm to base $2$'' and $\ln\left(\cdot\right)$ stands for the natural logarithm. Let $(\cdot)^T$ denote the transpose operation of the vector. Let $\mathbb{E}(Z)$ denote the expectation of the random variable $Z$ respectively. Furthermore, ``$\otimes$" is the Kronecker product and ``$\oplus$" denotes the module-$2$ operation.

\subsection{Encoding and Decoding of Polar Codes}
A B-DMC is modeled as $W: \mathcal{X}\to \mathcal{Y}$ with input alphabet $\mathcal{X}= \{0,1\}$ and output alphabet $\mathcal{Y}$ and the corresponding channel transition probabilities are defined as $W(y|x)$, $x\in \mathcal{X}$ and $y\in \mathcal{Y}$. Ar{\i}kan introduced the idea of channel polarization \cite{Polarcode_Arikan} whereby $N=2^n$ independent uses of B-DMC $W$ can be transformed into a group of polarized channels $W_N^{(i)}: \mathcal{X}\to \mathcal{Y} \times \mathcal{X}^{i-1}$, $i\in \llbracket 1,N \rrbracket$.

In order to construct an $(N,K)$ polar code, the set of polarized channels with high-reliability, denoted $\mathcal{A}$ and named information set ($|\mathcal{A}|=K$), is selected to carry information bits and the rest polarized channels belonging to the complement set $\mathcal{A}^c$ are assigned the fixed values, named frozen bits ($|\mathcal{A}^c|=N-K$). So the codeword $x_1^N$ of polar codes can be generated by a binary source block $u_1^N$ consisting of $K$ information bits and $N-K$ frozen bits as follows,
\begin{equation}\label{equation1}
x_1^N=u_1^N{\bf{F}}_N.
\end{equation}
Here, the $N$-dimension generator matrix\footnote{The generator matrix in the initial form is composed of the matrix ${\bf{F}}_N$ and the bit-reversal matrix \cite{Polarcode_Arikan}. In fact, the bit-reversal operation does not affect the reliability of the polarized channel. Hence, in this paper, we only use the matrix $\mathbf{F}_N$ as the equivalent generator matrix which is also the polar coding form in 5G standard \cite{5GNR_38212}.}, ${{\bf{F}}_N}$, can be recursively defined as ${{\bf{F}}_N} = {\bf{F}}_2^{ \otimes n}$, where ``$^{\otimes n}$'' denotes the $n$-th Kronecker product and ${{\bf{F}}_2} = \left[ { \begin{smallmatrix} 1 & 0 \\ 1 &  1 \end{smallmatrix} } \right]$ is the $2\times2$ kernel matrix.

The classic decoding algorithm of polar codes is the SC decoding algorithm \cite{Polarcode_Arikan} with a low complexity $O(N\log N)$. However, the error performance of SC decoding for the finite code length is not satisfied. Thereafter, many improved SC decoding algorithms, such as successive cancellation list (SCL) \cite{SCL_Tal}, successive cancellation stack (SCS) \cite{SCS_Niu}, successive cancellation hybrid (SCH) \cite{SCH_Chen}, successive cancellation priority (SCP) \cite{SCH_Guan}, and CRC aided (CA)-SCL/SCS \cite{SCL_Tal,CASCL_Niu,ASCL_Li,Survey_Niu} decoding were proposed to improve the performance of polar codes with an acceptable complexity.

\subsection{Polar Spectrum}
In the seminal paper \cite{Polarcode_Arikan}, the reliability of each polarized channel is evaluated by the Bhattacharyya parameter whereas this metric is only approximate for many B-DMCs. Then, many improved construction methods were proposed to estimate the error performance of polar codes, such as, density evolution (DE) \cite{DE_Mori}, Tal-Vardy algorithm \cite{Tal_Vardy} and Gaussian approximation (GA) \cite{GA_Trifonov}. Nevertheless, all these construction methods based on iterative calculation cannot explicitly reveal the characteristic of weight distribution of polar codes.

In \cite{PolarSpectrum_Niu}, we found that the reliability of each polarized channel $W_N^{(i)}$ is associated with two sets of specific codewords, named subcode and polar subcode respectively. They are defined as follows.
\begin{definition}\label{definition1}
Given the code length $N$, the $i$-th subcode $\mathbb{C}_N^{(i)}$ is defined as a set of codewords, that is,
\begin{equation}
\mathbb{C}_N^{(i)}\triangleq \left\{{\bf{c}}:{\bf{c}}=\left(0_1^{(i-1)},u_i^N\right){\bf{F}}_N,\forall u_i^N \in \mathcal{X}^{N-i+1} \right\}.
\end{equation}
Furthermore, one subset of the subcode $\mathbb{C}_N^{(i)}$, namely the polar subcode $\mathbb{D}_N^{(i)}$, can be defined as
\begin{equation}
\mathbb{D}_N^{(i)}\triangleq \left\{{\bf{c}}^{(1)}:{\bf{c}}^{(1)}=\left(0_1^{(i-1)},1,u_{i+1}^N\right){\bf{F}}_N,\forall u_{i+1}^N \in \mathcal{X}^{N-i} \right\}.
\end{equation}
\end{definition}

We introduced a new concept, named polar spectrum, to indicate the structure feature of the polar subcode. The polar spectrum is defined as the weight distribution set $\left\{A_N^{(i)}(d)\right\},d\in \llbracket 1, N \rrbracket $, where $d$ is the Hamming weight of nonzero codeword and the polar weight enumerator $A_N^{(i)}(d)$ enumerates the codewords of weight $d$ in the codebook $\mathbb{D}_N^{(i)}$.

Therefore, by using the polar spectrum, the error probability of the polarized channel $W_N^{(i)}$ is upper bounded by
\begin{equation}\label{equation2}
P\left(W_N^{(i)}\right)\leq \sum \limits_{d=1}^N A_N^{(i)}(d) P_N^{(i)}(d),
\end{equation}
where $P_N^{(i)}(d)$ is the pairwise error probability (PEP) between the all-zero codeword and the codeword with the Hamming weight $d$.

The previous evaluation methods of the error probability of the polarized channel, such as Bhattacharyya parameter, DE, Tal-Vardy or GA algorithm, cannot direct reveal the relation between the error performance and weight distribution of polar subcode. On the contrary, due to the analytical form, the union bound of the error probability (\ref{equation2}) bridges the gap between the error performance of polar codes and the structure feature of polar subcode. So the construction methods based on polar spectrum can present more constructive merits than the traditional methods.

\section{Performance Analysis based on Split Polar Spectrum}
\label{section_III}
In this section, the error performance of polar codes in the block fading channel is analyzed by using split polar spectrum. First, the signal model of the block fading channel is investigated. Then we derive the upper bound on the error probability of the polarized channel under the block Rayleigh fading channels and provide the design criteria of polar codes. Finally, the corresponding BLER upper bounds of the polar codes is also deduced.

\subsection{Signal Model of Block Fading Channel}
The communication system of polar coded diversity is depicted in Fig. \ref{PCD_system}. The source block $\mathbf{u}=u_1^N=\left(u_1,u_2,...,u_N\right)$ with $K$ information bits is encoded into a codeword $\mathbf{x}=x_1^N=\left(x_1,x_2,...,x_N\right)$. After a mapping operation, the bit sequence of the codeword is transformed into a transmitted signal vector $\mathbf{s}$ and further divided into $L$ equal-length blocks, that is, $\mathbf{s}=\left(\mathbf{s}_1,\mathbf{s}_2,...,\mathbf{s}_L\right)$.
So the code length satisfies $N=LM$, that is, the signal vector is composed of $L$ blocks and each block has $M$ elements.

\begin{figure}[h]\label{Figure1}
\setlength{\abovecaptionskip}{0.cm}
\setlength{\belowcaptionskip}{-0.cm}
  \centering{\includegraphics[scale=0.9]{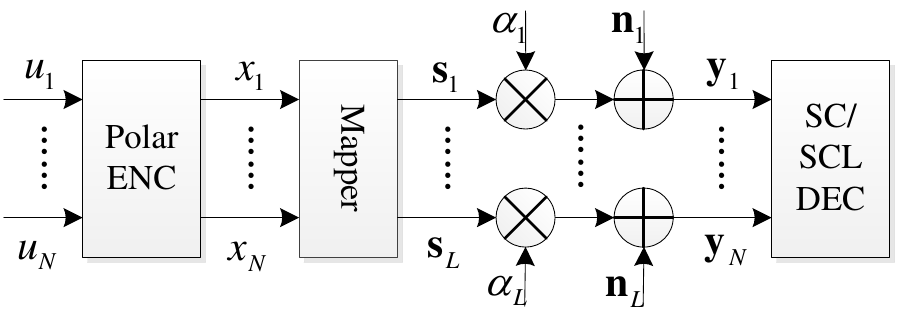}}
  \caption{Polar coded diversity communication system.}\label{PCD_system}
\end{figure}

When the signal vector $\mathbf{s}$ is transmitted over the block fading channel, each element $s_{l,k}$ (${l\in \llbracket 1,L \rrbracket, k\in \llbracket 1,M \rrbracket}$) in block $\mathbf{s}_l$ undergoes an independent block-fading $\alpha_l$ and is added a white Gaussian noise sample $n_{l,k}$. Finally, the received vector $\mathbf{y}=\left(\mathbf{y}_1,\mathbf{y}_2,...,\mathbf{y}_L\right)$ is sent to the polar decoder and decoded by the classic algorithms, such as, SC or SCL decoding.
So we can write the signal model as follows
\begin{equation}\label{equation3}
\mathbf{y}_l=\alpha_l \mathbf{s}_l+\mathbf{n}_l,l\in \llbracket 1,L \rrbracket
\end{equation}
where $\mathbf{y}_l=\left(y_{l,1},y_{l,2},...,y_{l,M}\right)$ is the $l$-th received block, $\mathbf{s}_l=\left(s_{l,1},s_{l,2},...,s_{l,M}\right)$ is the $l$-th transmitted block, and the additive noise vector $\mathbf{n}_l$ has zero-mean complex Gaussian elements each with variance $\sigma_n^2=N_0/2$. The element in the transmitted signal vector $\mathbf{s}$ is a BPSK signal, that is, $s_{l,k}=\pm\sqrt{E_s}$, where $E_s$ is the symbol energy. Further, we define $\frac{E_s}{N_0}$ as the symbol signal-to-noise ratio (SNR) and $N_0$ is the single-sided power spectral density of the additive white noise.

Especially, the amplitude gain vector $\bm{\alpha}=\left(\alpha_1,\alpha_2,...,\alpha_L\right)$ is composed of fading coefficients $\alpha_l$ which are constant during the transmission of a block of $M$ bits and independent from block to block. In this paper, we mainly consider the Rayleigh fading with the probability density function (pdf) of $\alpha_l$ given by \cite{Book_Proakis}
\begin{equation}\label {equation4}
f_{\alpha_l}(\alpha)=\frac{\alpha}{\sigma^2}\exp\left[-\frac{\alpha^2}{2\sigma^2}\right],
\end{equation}
where the average power of the fading envelope $\alpha_l$ obeys $\mathbb{E}(\alpha_l^2)=2\sigma^2=1$.

In Fig. \ref{PCD_system}, the mapper $\psi_L$ defines a memoryless mapping from the codebook to $L$-ary Cartesian product of real spaces $\mathcal{R}^M$
\begin{equation}
\psi_L: \mathcal{X}^N\to\left(\mathcal{R}^M\right)^L.
\end{equation}
This mapping includes interleaving, modulation and splitting into $L$ blocks. We focus on two mapping operations in this paper: \textbf{block mapping} and \textbf{random mapping}. For the first mapping, no interleaving is applied and the coded bits are directly mapped into BPSK modulated signals, that is, $s_{l,k}=\sqrt{E_s}\left(1-2x_{l,k}\right)$. The analysis and construction of polar codes under this block mapping are discussed in Section \ref{section_III} and Section \ref{section_IV}. On the other hand, the second mapping includes a uniform interleaver $\Pi(\cdot)$ and the signal can be written by $s_{l,k}=\sqrt{E_s}\left(1-2x_{\Pi(s,t)}\right)$. The corresponding analysis and construction of polar codes will be addressed in Section \ref{section_V} and Section \ref{section_VI}. Furthermore, we use $\gamma_l=\mathbb{E}(\alpha_l^2)\frac{E_s}{N_0}$ to denote the SNR of received block $l$. Without loss of generality, the average power of fading envelope can be normalized to one. So we define the average SNR as $\gamma=\mathbb{E}(\alpha^2)\frac{E_s}{N_0}=\frac{E_s}{N_0}$.

In this paper, we assume the perfect channel state information (CSI) is known by the receiver. Thus, given the fading coefficient $\alpha_l$, the block Rayleigh fading channel is modeled as $W_{M|\alpha_l}: \mathcal{X}^M\to \mathcal{Y}^M$ with input alphabet $\mathcal{X}^M= \{0,1\}^M$ and output alphabet $\mathcal{Y}^M\in \mathcal{R}^M$. Then, the conditional channel transition probabilities can be written as
\begin{equation}
W_M\left(\mathbf{y}_l\left|\mathbf{x}_l,\alpha_l\right.\right)=\frac{1}{\left(\pi N_0\right)^{M/2}}\prod\limits_{m=1}^{M}\exp\left\{-\frac{\left(y_{l,m}-\alpha_l s_{l,m}\right)^2}{N_0}\right\}.
\end{equation}
After the polarization transform, given the channel envelope vector $\bm{\alpha}$, the conditional transition probabilities of the synthetic channel $W_N$ satisfy
\begin{equation}\label{equation5}
W_N\left(\mathbf{y}\left|\mathbf{x},\bm{\alpha}\right.\right)=\frac{1}{\left(\pi N_0\right)^{N/2}}\prod\limits_{l=1}^{L}\prod\limits_{m=1}^{M}\exp\left\{-\frac{\left(y_{l,m}-\alpha_l s_{l,m}\right)^2}{N_0}\right\}.
\end{equation}
Correspondingly, we use $W_N^{(i)}(\bm{\alpha})=W_N^{(i)}\left(y_1^N,u_1^{i-1}\left|u_i,\bm{\alpha}\right.\right)$ to denote the conditional transition probabilities of the polarized channel under the block fading channel.

\subsection{Block Error Rate Bound Based on Split Polar Spectrum}
The error performance of polar codes in the block Rayleigh fading channel is dominated by the split polar spectrum and the PEP. First, we derive the upper bound of PEP as the following theorem.
\begin{theorem}\label{theorem1}
Assuming the transmission vector is ${\bf{c}}^{(0)}=0_1^N$ and the decision vector is ${\bf{c}}^{(1)}$, if the $L$-partition block mapping is used, the block-wise Hamming distance vector between these two vectors is written as $\mathbf{d}=\left(d_1,d_2,...,d_L\right)$ where $d_l=d_H\left({\bf{c}}_l^{(0)},{\bf{c}}_l^{(1)}\right),l\in\llbracket 1,L\rrbracket$. So the PEP between ${\bf{c}}^{(0)}$ and ${\bf{c}}^{(1)}$ on the block Rayleigh fading channel can be bounded by
\begin{equation}\label{equation8}
P\left({\bf{c}}^{(0)}\to{\bf{c}}^{(1)}\right)\leq \prod\limits_{l=1}^{L} \frac{1}{1+d_l {E_s}/{N_0}}.
\end{equation}
\end{theorem}
\begin{proof}
Given the block mapping, the codeword ${\bf{c}}^{(0)}$ and ${\bf{c}}^{(1)}$ are mapped into the transmission vector ${\bf{s}}^{(0)}$ and the signal vector ${\bf{s}}^{(1)}$ respectively. Corresponding to the block partition of the codewords, the vectors $\mathbf{s}^{(0)}$ and $\mathbf{s}^{(1)}$ are also divided into $L$ blocks, that is, $\mathbf{s}^{(j)}=\left(\mathbf{s}_1^{(j)},\mathbf{s}_2^{(j)},...,\mathbf{s}_L^{(j)}\right), j=0,1$. For the block fading channel, we assume the channel envelope vector $\bm{\alpha}$ is available for the decoder. If a pairwise error occurs, that means the Euclidian distance between the received vector $\mathbf{y}$ and the transmitted vector $\mathbf{s}^{(0)}$ is lager than that between the received vector $\mathbf{y}$ and the signal vector $\mathbf{s}^{(1)}$. So we can obtain the inequality $\sum_{l=1}^{L} {\left\| {\mathbf{y}_l}- {\alpha_l}\mathbf{s}_l^{(0)} \right\|^2} >\sum_{l=1}^{L} \left\| \mathbf{y}_l-\alpha_l\mathbf{s}_l^{(1)} \right\|^2$.

Substituting (\ref{equation4}) into this inequality, we have $\!\sum_{l=1}^{L}{\left\| \mathbf{n}_l\right\|^2}\!>\!\sum_{l=1}^{L}\!\left\| \alpha_l \left(\mathbf{s}_l^{(0)}-\mathbf{s}_l^{(1)}\right)+{\mathbf{n}_l} \right\|^2\!$. After some manipulations, the decision region of the pairwise error event can be written as $\!\mathcal{H}=\left\{\mathbf{n}:\sum\limits_{l=1}^{L}{\mathbf{n}_l} \alpha_l \left(\mathbf{s}_l^{(0)}-\mathbf{s}_l^{(1)}\right)^T\!<\!- \frac{1}{2}\sum\limits_{l=1}^{L}\!\left\| \alpha_l \left(\mathbf{s}_l^{(0)}- \mathbf{s}_l^{(1)}\right)\right\|^2\right\}\!$. So given the channel envelope vector $\bm{\alpha}$, the conditional PEP can be given by
\begin{equation}\label{equation9}
\begin{aligned}
  P\left({\mathbf{c}}^{(0)}\to{\mathbf{c}}^{(1)}\left|\bm{\alpha}\right.\right)&= {\int { \cdots \int  } }_{\mathcal{H}} {W_N}\left( \mathbf{y}\left| {\mathbf{s}}^{(0)},\bm{\alpha} \right. \right)d{\mathbf{y}} \\
   &= \prod_{l=1}^{L} Q\left[ \sqrt {\frac{\alpha_l^2}{2N_0}\left\| \mathbf{s}_l^{(0)}-\mathbf{s}_l^{(1)} \right\|^2} \right] \\
   &\leq \prod_{l=1}^{L} \exp\left( -\alpha_l^2 d_l\frac{E_s}{N_0} \right).
\end{aligned}
\end{equation}
In the second equality of (\ref{equation9}), $Q(x)=\frac{1}{\sqrt{2\pi}}\int_x^{\infty} e^{-t^2/2}dt$ is the tail distribution function of the standard normal distribution.
Furthermore, considering the relation between the Euclidian distance and the Hamming distance, we have $\left\|\mathbf{s}_l^{(0)}-\mathbf{s}_l^{(1)}\right\|^2=4 E_s \left\| \mathbf{c}_l^{(0)}-\mathbf{c}_l^{(1)}\right\|^2=4 E_s d_l$. Then by using the inequality $Q(x)\leq e^{-\frac{x^2}{2}}$, we obtain the third inequality of (\ref{equation9}).

Since $\alpha_l$ is the Rayleigh-distributed random variable, we conclude that $\alpha_l^2$ in Eq. (\ref{equation9}) is a central chi-square-distributed random variable with $2$ degrees of freedom \cite{Book_Proakis}. Hence the pdf of the fading block SNR $\gamma_l=\alpha_l^2 E_s/N_0$ is a negative exponential distribution as the following
\begin{equation}
g\left(\gamma_l\right)=\frac{1}{\gamma}e^{-\gamma_l/\gamma}.
\end{equation}

For the block Rayleigh fading channel, taking the expectation over the fading block SNRs, the average PEP can be derived as
\begin{equation}\label{equation11}
\begin{aligned}
\mathbb{E}_{\bm\alpha}\left[P\left({\mathbf{c}}^{(0)}\to{\mathbf{c}}^{(1)}\left|\bm{\alpha}\right.\right)\right]&=\prod_{l=1}^{L}\int\nolimits_0^{\infty}e^{-\gamma_l d_l }g(\gamma_l)d\gamma_l \\
&=\prod_{l=1}^{L} \frac{1}{1+\gamma d_l}.
\end{aligned}
\end{equation}
\end{proof}

Considering the derivation of Theorem \ref{theorem1}, it follows that PEP $P\left({\bf{c}}^{(0)}\to {\bf{c}}^{(1)}\right)$ is determined by the weight distribution vector $\bf{d}$ of polar codeword. So we introduce the following concept to indicate the number of codewords with the given weight distribution vector.
\begin{definition}\label{definition1}
The split polar spectrum of the polar subcode $\mathbb{D}_N^{(i)}$, also termed as split polar weight distribution, is defined as the weight distribution set $\left\{A_N^{(i)}\left(d_1,d_2,...,d_L\right)\right\},d_l\in \llbracket 1, M \rrbracket, l\in\llbracket 1,L\rrbracket $, where $\sum_{l=1}^{L}{d_l}=d$ is the Hamming weight of nonzero codeword. The split polar weight enumerator $A_N^{(i)}(\bf{d})$ enumerates the number of the given weight partition vector ${\bf{d}}=\left(d_1,d_2,...,d_L\right)$ for codebook $\mathbb{D}_N^{(i)}$.
\end{definition}

Using the average PEP and split polar spectrum, we can derive the union bound of the average error probability of the polarized channel $W_N^{(i)}$ by the following corollary.
\begin{corollary}\label{corollary1}
Given the polar subcode ${\mathbb{D}}_N^{(i)}$, the error probability of the polarized channel $W_N^{(i)}$ can be upper bounded as
\begin{equation}\label{equation12}
P\left(W_N^{(i)}\right)\leq \sum \limits_{\sum_{l=1}^{L}{d_l}=d_{min}^{(i)}}^N A_N^{(i)}({\mathbf{d}}) \prod\limits_{l=1}^{L} \frac{1}{d_l {E_s}/{N_0}+1},
\end{equation}
where $d_{min}^{(i)}$ denotes the minimum Hamming distance of polar subcode ${\mathbb{D}}_N^{(i)}$.
\end{corollary}
Note that this upper bound is the union bound of the error probability and the summation should be taken over all the weight partitions of the codeword weight.

By using Eq. (\ref{equation2}), we have the following corollary to evaluate the BLER upper bound in the block Rayleigh fading channel.
\begin{corollary}\label{corollary2}
Given the fixed configuration $(N,K,\mathcal{A})$, the block error probability of polar code using SC decoding in the block Rayleigh fading channel is upper bounded by
\begin{equation}\label{equation13}
P_{e}(N,K,\mathcal{A})\leq \sum\limits_{i\in \mathcal{A}} \sum \limits_{\sum_{l=1}^{L}{d_l}=d_{min}^{(i)}}^N A_N^{(i)}({\mathbf{d}}) \prod\limits_{l=1}^{L} \frac{1}{d_l {E_s}/{N_0}+1}.
\end{equation}
\end{corollary}

Compared with the previous works of polar coding in block fading channel \cite{BlockFading_Boutros,BlockFading_Bravo-Santos,BlockFading_Si,BlockFading_Liu}, the upper bounds in (\ref{equation12}) and (\ref{equation13}) provide an analytical form for the error probability of polar codes in the block Rayleigh fading channel. These bounds are mainly determined by the spit polar spectrum of the selected polar subcodes, that is to say, the weight partition vectors dominate the average PEP and the split polar weight enumerators indicate the number of weight partitions. Therefore, for the construction of polar codes in block fading channel, these upper bounds are more attractive than the traditional methods. Next, we will further discuss the design criteria of polar codes on the block Rayleigh fading channels.

\subsection{Design Criteria on the Block Fading Channel}\label{DesignCriteria}
By Theorem \ref{theorem1}, we investigate the upper bound of PEP in the case of high SNR. In this case, suppose $\forall l, d_l\neq0$ and $E_s/N_0\gg 1$, the PEP can be further bounded by
\begin{equation}\label{equation14}
\begin{aligned}
P\left({\bf{c}}^{(0)}\to{\bf{c}}^{(1)}\right)&\leq \prod\limits_{l=1}^{L} \frac{1}{d_l {E_s}/{N_0}}
                                                            = \left(\prod\limits_{l=1}^{L}{d_l}\right)^{-1} \left(\frac{E_s}{N_0}\right)^{-L}.
\end{aligned}
\end{equation}

Obviously, considering the power of SNR in the denominator of Eq. (\ref{equation14}), we observe that the full diversity $L$ can be achieved if the block Hamming distance $d_l$ is none zero. Furthermore, we find that $\prod_{l=1}^{L}{d_l}$, named product distance, is a coding advantage over an uncoded system operating with the same diversity order. Hence, given the fixed codeword weight $d$, in order to minimize the PEP, we should maximize the product distance, which is expressed as the following optimization problem,
\begin{equation}\label{equation15}
\begin{aligned}
&\max_{\mathbf{d}=\left(d_1,d_2,...,d_L\right)}\prod_{l=1}^{L}{d_l}\\
s.t. &\forall {l}, \text{   } d_l > 0, \sum_{l=1}^{L}{d_l}=d.
\end{aligned}
\end{equation}
Theoretically, due to the integer constraint of Hamming distance, this integer-programming problem is only solved by a brute-force searching. However, if the integer constraint is relaxed, we can obtain the following lemma for the approximate optimal solution.
\begin{lemma}\label{lemma1}
The optimal solution of (\ref{equation15}) can be achieved when the codeword weight is uniformly partitioned, that is, $\forall l, d_l=d/L$.
\end{lemma}

Lemma \ref{lemma1} can be easily proved by using the method of Lagrangian multiplier. Based on this lemma, we obtain the minimum upper bound of PEP as follows,
\begin{equation}\label{equation16}
P\left({\bf{c}}^{(0)}\to{\bf{c}}^{(1)}\right) \leq \left(\frac{d}{L}\frac{E_s}{N_0}\right)^{-L}.
\end{equation}

Thus from the above analysis, we arrive at the following design criteria.

\textit{Design Criteria for Polar Codes in Rayleigh Block Fading Channel:}
\begin{itemize}
  \item \textit{Full Diversity Criterion:} Given the block Rayleigh fading channel with $L$ blocks, if the polar subcode $\mathbb{D}_N^{(i)}$ associated with the $i$-th polarized channel has nonzero weight partition, that is, the corresponding product distance satisfies $\prod \nolimits_{l=1}^{L} d_l \neq 0$, this polar subcode will achieve a diversity of $L$. Consequently, for a fixed configuration $(N,K,\mathcal{A})$, if all the selected polarized channels in the set $\mathcal{A}$ satisfy this condition, the full diversity gain can be obtained. This criterion is the primary target of polar code design.
  \item \textit{Product Distance Criterion:} In order to minimize the error probability of the polarized channel or BLER upper bound of polar codes, the corresponding product distance should be maximized. Conceptually, uniform partition of codeword weight is optimal in the case of high SNR. However, considering the influence of the split polar weight enumerator and the integer constraint of the block-wise distance, we still need to carefully optimize this metric. In this sense, the product distance criterion is the second target of polar code design.
\end{itemize}

\section{Polar Coding Based on Block Mapping}
\label{section_IV}
In this section, we focus on the construction of polar codes under the block mapping. First, we derive the constructive metric based on the logarithmic version of the upper bound, named polarized diversity weight (PDW). Second, for the special case of $L=2$, we provide the enumeration algorithm of split polar spectrum based on the general MacWilliams identities.
\subsection{Construction Metrics}
Now we investigate the construction of polar codes in the block mapping. By using Corollary \ref{corollary1}, the upper bound of the channel error probability can be used as a reliability metric to sort all the polarized channels. In order to facilitate implementation in the practical application, the logarithmic form of the upper bound is more desirable.
Since the approximation of high-SNR will introduce a large deviation and cannot accurately indicate the reliability of the polarized channel, we mainly consider the upper bound on error probability of the polarized channel in the case of low-SNR and have the following theorem.
\begin{theorem}\label{theorem3}
Given the block Rayleigh fading channel and the code length $N$, for the low-SNR, the upper bound on the error probability of the polarized channel can be approximated by
\begin{equation}
P\left(W_N^{(i)}\right)\lesssim \sum \limits_{\sum_{l=1}^{L}{d_l}=d_{min}^{(i)}}^N A_N^{(i)}({\mathbf{d}}) \prod\limits_{l=1}^{L} \frac{e^{-\left(d_l E_s/N_0+1\right)}}{1+d_l {E_s}/{N_0}-V_N^{(i)}(\mathbf{d})},
\end{equation}
where $V_N^{(i)}(\mathbf{d})=\frac{1}{L}\ln A_N^{(i)}(\mathbf{d})$.
\end{theorem}
\begin{proof}
By using Theorem \ref{theorem1} and Corollary \ref{corollary1}, the upper bound on the error probability of polarized channel can be rewritten as
\begin{equation}\label{equation20}
\begin{aligned}
P\left(W_N^{(i)}\right)&\leq \sum\limits_{\sum_{l=1}^{L}{d_l}=d_{min}^{(i)}}^{N} \prod\limits_{l=1}^{L} {\left[A_N^{(i)}({\mathbf{d}})\right]}^{1/L} \int\nolimits_0^{\infty} e^{-\gamma_l d_l }g(\gamma_l)d\gamma_l\\
& =\sum\limits_{\sum_{l=1}^{L}{d_l}=d_{min}^{(i)}}^{N} \prod\limits_{l=1}^{L}\int\nolimits_0^{\infty} e^{V_N^{(i)}(\mathbf{d})-\gamma_l d_l }\frac{1}{\gamma}e^{-\gamma_l/\gamma}d\gamma_l.
\end{aligned}
\end{equation}
The inner integration in (\ref{equation20}) can be further decomposed into two parts with distinct integration intervals as follows,
\begin{equation}
\begin{aligned}
&\int\nolimits_0^{\infty} \frac{1}{\gamma}e^{V_N^{(i)}(\mathbf{d})-\gamma_l \left(d_l+\frac{1}{\gamma}\right) }d\gamma_l=I_1+I_2\\
&=\int\nolimits_0^{\gamma} \frac{1}{\gamma}e^{V_N^{(i)}(\mathbf{d})-\gamma_l \left(d_l+\frac{1}{\gamma}\right) }d\gamma_l+
\int\nolimits_{\gamma}^{\infty} \frac{1}{\gamma}e^{V_N^{(i)}(\mathbf{d})-\gamma_l \left(d_l+\frac{1}{\gamma}\right) }d\gamma_l.
\end{aligned}
\end{equation}
When the SNR is very low, we have $\gamma\approx 0$. So the first integration $I_1$ can be omitted, i.e., $I_1\approx 0$. Thus, the upper bound is mainly dominated by the second integration $I_2$, which can be further bounded by
\begin{equation}\label{equation22}
\begin{aligned}
I_2 &\leq \int\nolimits_{\gamma}^{\infty} \frac{1}{\gamma}e^{V_N^{(i)}(\mathbf{d})\frac{\gamma_l}{\gamma}-\gamma_l \left(d_l+\frac{1}{\gamma}\right) }d\gamma_l\\
&=\frac{\exp{\left(-\left(\gamma d_l+1-{V_N^{(i)}(\mathbf{d})}\right)\right)}}{\gamma d_l+1-{V_N^{(i)}}(\mathbf{d})}.
\end{aligned}
\end{equation}
In the derivation of (\ref{equation22}), we use the inequality $\frac{\gamma_l}{\gamma}\geq 1$.
\end{proof}

Note that in Theorem \ref{theorem3}, the SNR should satisfy the condition $ \forall l, d_l E_s/N_0+1-{V_N^{(i)}}(\mathbf{d})>0$.
For the case of low SNR, by using Theorem \ref{theorem3}, the logarithmic form of the upper bound can be written as
\begin{equation}
\begin{aligned}
&\ln \left\{\sum \limits_{\sum_{l=1}^{L}{d_l}=d_{min}^{(i)}}^N A_N^{(i)}({\mathbf{d}}) \prod\limits_{l=1}^{L} \frac{e^{-\left(d_l E_s/N_0+1\right)}}{1+d_l {E_s}/{N_0}-V_N^{(i)}(\mathbf{d})}\right\}\\
&\begin{aligned}\propto \max_{\mathbf{d},\sum_{l=1}^{L}{d_l}=d_{min}^{(i)}}
&\left\{\ln {A_N^{(i)}(\mathbf{d})}- d_{min}^{(i)}\frac{E_s}{N_0} \right.\\
&\left.-\sum_{l=1}^{L}\ln \left[{1+d_l \frac{E_s}{N_0}-V_N^{(i)}(\mathbf{d})}\right]\right\}.
\end{aligned}
\end{aligned}
\end{equation}
Here the approximation $\ln\left( \sum\limits_{k} e^{a_k} \right)\approx \max\limits_{k}\{a_k\}$ is used and some constants are omitted. When $d_l \frac{E_s}{N_0}-V_N^{(i)}(\mathbf{d})$ is sufficiently small, by using the approximation $\ln(1+x)\approx x$, we obtain the PDW as below.

\begin{metric}\label{theorem4}
Given the code length $N$, the reliability of the polarized channel can be sorted by the logarithmic version metric, namely, the polarized diversity weight (PDW), that is,
\begin{equation}\label{equation24}
PDW_N^{(i)}=\max_{\mathbf{d},\sum_{l=1}^{L}{d_l}=d_{min}^{(i)}} \left[L_N^{(i)}(\mathbf{d})-d_{min}^{(i)} \frac{E_s}{N_0}\right],
\end{equation}
where $L_N^{(i)}(\mathbf{d})=\ln A_N^{(i)}(\mathbf{d})$ is the logarithmic version of spilt polar weight enumerator.
\end{metric}

\begin{example}\label{example1}
If $L=1$, the signal vector $\mathbf{s}$ will undergo the same fading. In this case, the split polar spectrum will degrade to polar spectrum, that is, $L_N^{(i)}(\mathbf{d})\to L_N^{(i)}(d)$. Thus, PDW is written as
\begin{equation}
PDW_N^{(i)}= L_N^{(i)}\left(d_{min}^{(i)}\right)- d_{min}^{(i)} \frac{E_s}{N_0}.
\end{equation}

If $L=2$, the signal vector $\mathbf{s}$ will undergo two independent block fading envelopes. Suppose the weight partition is $\left(d_1,d_2\right)$, so PDW can be written as
\begin{equation}
PDW_N^{(i)}=\max_{d_1,d_2,d_1+d_2=d_{min}^{(i)}} \left[L_N^{(i)}(d_1,d_2)- d_{min}^{(i)} \frac{E_s}{N_0}\right].
\end{equation}
\end{example}

\begin{remark}
For the traditional constructions \cite{BlockFading_Boutros,BlockFading_Bravo-Santos,BlockFading_Si,BlockFading_Liu} in block fading channels, the reliability metrics, such as outage capacity or Bhattacharyya parameter, are iteratively evaluated with a medium or high complexity. On the contrary, since PDW is mainly determined by the split polar spectrum and the average SNR, this metric has explicitly analytical structure and indicates the reliability order of the polarized channel in block Rayleigh fading channels. Particularly, if the split polar spectrum can be pre-calculated by the enumerating algorithm based on the general MacWilliams identities (described in the next subsection), the complexity of this construction in block mapping is linear in terms of $O(N)$, which is much lower than that of the former algorithms. Furthermore, this metric can also be transformed into a channel-independent construction by selecting a fixed $E_s/N_0$. In this sense, due to the explicit analyticity and low-complexity construction, the proposed construction based on PDW is attractive for the practical application.
\end{remark}
\subsection{Enumeration of Split Polar Spectrum for $L=2$}
We can follow the idea of \cite{PolarSpectrum_Niu} and enumerate the split polar spectrum. However, for the general $L$-split polar spectrum, the computational complexity is very high. Hence, we only consider the special case $L=2$ whereby $2$-split polar spectrum $A_N^{(i)}(d_1,d_2)$ is enumerated.

Let $i \in \llbracket N/2+1, N\rrbracket$ denote the row index of the matrix $\mathbf{F}_N$. By the definition in \cite{PolarSpectrum_Niu}, the generator matrices of subcode $\mathbb{C}_N^{(i)}$ and $\mathbb{C}_N^{(N+2-i)}$ satisfy $\mathbf{G}_{\mathbb{C}_N^{(i)}}=\mathbf{F}_N(i:N)$ and $\mathbf{G}_{\mathbb{C}_N^{(N+2-i)}}=\mathbf{F}_N(N+2-i:N)$ respectively. So we prove that these two subcodes are dual \cite[Theorem 5]{ PolarSpectrum_Niu}, that is, $\mathbb{C}_N^{(N+2-i)}=\mathbb{C}_N^{\bot(i)}$.

Let $S_N^{(i)}(j,k)(0\leq j,k\leq N/2)$ denote the split weight enumerators of the subcode $\mathbb{C}_N^{(i)}$, where the $d$-weight codeword of codebook $\mathbb{C}_N^{(i)}$ is partitioned into two blocks with the weights $j$ and $k$. Similarly, $S_N^{\bot( i)}(j,k)$ denote the weight enumerators of the dual code $\mathbb{C}_N^{\bot (i)}$.

The linear relations between the split weight distributions of a linear code and its dual can also be determined by the general MacWilliams identities \cite{SplitWeight}. These identities provide a simple method to calculate the split weight distribution.

\begin{theorem}\label{theorem5}
Given the subcode $\mathbb{C}_N^{(i)}$ and its dual $\mathbb{C}_N^{\bot(i)}=\mathbb{C}_N^{(N+2-i)}$, the split weight enumerators $S_N^{(i)}(j,k)$ and $S_N^{\bot(i)}(j,k)$ satisfy the following general MacWilliams identities \cite{SplitWeight}
\begin{equation}\label{equation32}
\begin{aligned}
&\sum_{j=0}^{N/2}\sum_{k=0}^{N/2}
\!\left(\!\begin{array}{c} \!N/2-j\!\\\!s \end{array}\! \right)\!
\!\left(\!\begin{array}{c} \!N/2-k\!\\\!t \end{array} \!\right)\!
S_N^{\bot(i)}(j,k)=\\
&2^{i-1-s-t} \sum_{j=0}^{N/2}\sum_{k=0}^{N/2}
\left(\!\begin{array}{c} \!N/2-j\!\\\!N/2-s\! \end{array}\!\right)
\left(\!\begin{array}{c} \!N/2-k\!\\\!N/2-t\! \end{array}\!\right)
S_N^{(i)}{( j,k)},
\end{aligned}
\end{equation}
where $s,t \in \llbracket 0,N/2\rrbracket$.
\end{theorem}
By solving these $(N/2+1)^2$ linear equations, we can calculate the split weight distribution of one subcode and its dual.

\begin{lemma}\label{lemma2}
Given the subcode $\mathbb{C}_N^{(i)}$, we have $\mathbb{C}_N^{(i)}=\mathbb{D}_N^{(i)}\bigcup \mathbb{C}_N^{(i+1)}$. Thus, the split weight enumerator and the split polar weight enumerators satisfy $S_N^{(i)}(j,k)=A_N^{(i)}(j,k)+S_N^{(i+1)}(j,k)$.
\end{lemma}
This lemma is similar to Proposition 6 in \cite{PolarSpectrum_Niu} and can be easily proved.

\begin{lemma}\label{lemma3}
For $\forall i\in\llbracket N/2+1,N\rrbracket$, the weight distribution vector $\mathbf{d}=(d_1,d_2)$ (related to the codeword weight $d$) of subcode $\mathbb{C}^{(i)}$ or polar subcode $\mathbb{D}^{(i)}$ is equally partitioned, that is, $d_1=d_2=d/2$.
\end{lemma}
\begin{proof}
For a subcode $\mathbb{C}_N^{(i)},i \in \llbracket N/2+1, N \rrbracket$, the source block $\mathbf{u}=\left(\mathbf{u}_1,\mathbf{u}_2\right)$ can be decomposed into $\mathbf{u}_1=0_1^{N/2}$ and $\mathbf{u}_2=\left(0_{N/2+1}^{i-1},b_i^{N}\right)$, where $b_i^{N}\in \{0,1\}^{N-i+1}$ is an arbitrary binary vector. Due to Plotkin's structure $[\bf{u}+\bf{v}|\bf{v}]$, the corresponding codeword $\mathbf{c}=\left(\mathbf{c}_1,\mathbf{c}_2\right)$ can be generated as $\mathbf{c}_1=\mathbf{u}_1\mathbf{F}_{N/2}+\mathbf{u}_2\mathbf{F}_{N/2}=\mathbf{u}_2\mathbf{F}_{N/2}$ and $\mathbf{c}_2=\mathbf{u}_2\mathbf{F}_{N/2}$. So we have $\mathbf{c}_1=\mathbf{c}_2$ and complete the proof.
\end{proof}

Lemma \ref{lemma3} means that half of the polarized channels ($N/2+1\leq i\leq N$) can achieve $2$-diversity order. However, since some zero-bit blocks may exist, the other half of the polarized channels ($1\leq i \leq N/2$) may not obtain the same diversity order.

\begin{lemma}\label{lemma4}
The split weight distribution of the subcode $\mathbb{C}_N^{(i)}$ is symmetric, that is, $S_N^{(i)}(j,k)=S_N^{(i)}(j,N/2-k)=S_N^{(i)}(N/2-j,k)=S_N^{(i)}(N/2-j,N/2-k)=S_N^{(i)}(k,j)$. Similarly, for the polar subcode $\mathbb{D}_N^{(i)}$, we also have $A_N^{(i)}(j,k)=A_N^{(i)}(j,N/2-k)=A_N^{(i)}(N/2-j,k)=A_N^{(i)}(N/2-j,N/2-k)=A_N^{(i)}(k,j)$.
\end{lemma}
\begin{proof}
Due to Plotkin's structure $[\bf{u}+\bf{v}|\bf{v}]$,  similar to Proposition 8 in \cite{PolarSpectrum_Niu}, we can easily obtain these conclusions.
\end{proof}

Based on Lemma \ref{lemma2}$\sim$\ref{lemma4} and the general MacWilliams identities, we design an enumeration algorithm to calculate the split polar spectrum. Suppose the (polar) weight distribution of (polar subcodes $\mathbb{D}_N^{(i)}$) subcodes $\mathbb{C}_N^{(i)}$ has been calculated by using Algorithm 1 in \cite{PolarSpectrum_Niu}, the following algorithm can enumerate the corresponding split polar spectrum.

\begin{algorithm}[h] \label{algorithm1}
\setlength{\abovecaptionskip}{0.cm}
\setlength{\belowcaptionskip}{-0.cm}
\caption{Iterative enumeration algorithm of split polar spectrum}
\KwIn {The weight distribution of all the subcodes with the code length $N$, $\left\{S_N^{(i)}(j),A_N^{(i)}(j):i=1,2,...,N,j=0,1,...,N\right\}$;}
\KwOut {The split polar spectrum of all the polar subcodes with the code length $2N$, $\left\{A_{2N}^{(l)}(j,k):l=1,2,...,2N,j,k=0,1,...,N\right\}$;}

Initialization $\forall j,k, A_{2N}^{(l)},(j,k)=0, S_{2N}^{(l)},(j,k)=0$\;
\For{$l = N+1 \to 2N $}
{
    if ($j=k$)
    {  Calculate the split polar spectrum of the polar subcodes with the code length $2N$, $A_{2N}^{(l)}(j,j)=A_{N}^{(l-N)}(j)$\;
        Calculate the split weight distribution of the subcodes with the code length $2N$, $S_{2N}^{(l)}(j,j)=S_{N}^{(l-N)}(j)$\;
    }
}
\For{$l = 2 \to N $}
{
    Solve the general MacWilliams Identities by (\ref{equation32}) and calculate the split weight distribution $S_{2N}^{(l)}(j,k)$\;
    Calculate the split polar spectrum $A_{2N}^{(l)}(j,k)=S_{2N}^{(l)}(j,k)-S_{2N}^{(l+1)}(j,k)$\;
}
\For{$j,k = 0 \to N $}
{
    Initialize the split weight distribution $S_{2N}^{(1)}(j,k)=\left( \begin{array}{*{20}{c}} N\\j\end{array}\right)\left( \begin{array}{*{20}{c}} N\\k\end{array}\right)$\;
    Calculate the split polar spectrum $A_{2N}^{(1)}(j,k)=S_{2N}^{(1)}(j,k)-S_{2N}^{(2)}(j,k)$
}

\end{algorithm}

Algorithm \ref{algorithm1} mainly includes two steps to enumerate the split polar spectrum. In the first step, when the code length is grown from $N$ to $2N$, we enumerate the split weight distribution and split polar spectrum in the case of $N+1 \leq l\leq 2N$. In this case, by using Lemma \ref{lemma3}, the subcode $\mathbb{C}_{2N}^{(l)}$ is consist of two identical component codes $\mathbb{C}_{N}^{(l-N)}$. Therefore, the split weight distribution of subcode $\mathbb{C}_{2N}^{(l)}$ is equipartition. And similar results are also concluded for the split polar spectrum of polar subcode $\mathbb{D}_{2N}^{(l)}$.

For the second step, we enumerate the split weight distribution and split polar spectrum in the case of $1 \leq l\leq N$. According to Theorem \ref{theorem5}, the subcode $\mathbb{C}_{2N}^{(l)}$ is the dual of subcode $\mathbb{C}_{2N}^{(2N+2-l)}$, we can solve the general MacWilliams identities to calculate the split weight distribution of $\mathbb{C}_{2N}^{(l)}$. Furthermore, by Lemma \ref{lemma2}, the split polar spectrum of $\mathbb{D}_{2N}^{(l)}$ is obtained.

The computational complexity of Algorithm \ref{algorithm1} is mainly determined by the solution of general MacWilliams identities. Due to the regular structure in the general MacWilliams identities, the split weight enumerators can be recursively calculated. Hence, given the code length $N$, the worst-case complexity of solving the general MacWilliams identities is $\chi_M(N)=(N/2+1)^4$. Since $N/2-1$ groups of those identities need to be calculated, the total computational complexity is $\chi_E(N)=(N/2-1)(N/2+1)^4$. Furthermore, using Lemma \ref{lemma4} (weight distribution symmetry), the worst-case complexity of enumeration algorithm can be further reduced to $\chi_E(N)=\frac{1}{4}\times(N/2-1)(N/2+1)^4\approx\frac{1}{128}N^5$. Thus the total complexity of Algorithm \ref{algorithm1} is $O(N^5)$.

\begin{example}
The split polar spectrum for the code length $N=16$ is partially shown in Table \ref{split_polar_spectrum}. In this example, due to the symmetric distribution of polar spectrum (Lemma \ref{lemma4}), some different weight partition vectors have the same split weight enumerators, e.g. $A_{16}^{(2)}(2,2)=A_{16}^{(2)}(2,6)=A_{16}^{(2)}(6,2)=A_{16}^{(2)}(6,6)=384$. As shown in this table, due to the duality relationship, $A_{16}^{(2)}$ and $A_{16}^{(16)}$, $A_{16}^{(3)}$ and $A_{16}^{(15)}$, etc. satisfy the general MacWilliams identities. Meanwhile, the subcode $\mathbb{C}_{16}^{(9)}$ is a self-dual code.
\end{example}

\begin{table}[tp]
\centering
\caption{Split polar spectrum example for $N=16$} \label{split_polar_spectrum}
\begin{tabular}{|c|c|c|}
\hline index $i$ & weight $\mathbf{d}$ & $A_{16}^{(i)}(\mathbf{d})$ \\
\hline   1 & \tabincell{c}{(1, 0), (1, 8), (7, 0), (7, 8)\\(0, 1), (8, 1), (0, 7), (8, 7)}& 8\\
\hline   1 & \tabincell{c}{(3, 0), (3, 8), (5, 0), (5, 8)\\(0, 3), (8, 3), (0, 5), (8, 5)}& 56 \\
\hline   1 & \tabincell{c}{(2, 1), (2, 7), (6, 1), (6, 7)\\(1, 2), (7, 2), (1, 6), (7, 6)}& 224\\
\hline   1 & \tabincell{c}{(4, 1), (4, 7), (1, 4), (7, 4)}& 560  \\
\hline   1 & \tabincell{c}{(3, 2), (3, 6), (5, 2), (5, 6)\\(2, 3), (6, 3), (2, 5), (6, 5)} &   1568  \\
\hline   1 & \tabincell{c}{(4, 3), (4, 5), (3, 4), (5, 4)} & 3920  \\
\hline   2 & \tabincell{c}{(2, 0), (2, 8), (6, 0), (6, 8)\\(0, 2), (8, 2), (0, 6), (8, 6)} &  16  \\
\hline   2 & \tabincell{c}{(4, 0), (4, 8), (0, 4), (8, 4)} &32  \\
\hline   2 & \tabincell{c}{(1, 1), (1, 7), (7, 1), (7, 7)\\} & 32  \\
\hline   2 & \tabincell{c}{(3, 1), (3, 7), (5, 1), (5, 7)\\(1, 3), (7, 3), (1, 5), (7, 5)} &    224  \\
\hline   2 & \tabincell{c}{(2, 2), (2, 6), (6, 2), (6, 6)} &  384  \\
\hline   2 & \tabincell{c}{(4, 2), (4, 6), (2, 4), (6, 4)} & 992  \\
\hline   2 & \tabincell{c}{(3, 3), (3, 5), (5, 3), (5, 5)} & 1568 \\
\hline   2 & \tabincell{c}{(4, 4)} & 2432  \\
\hline   3 & \tabincell{c}{(2, 0), (2, 8), (6, 0), (6, 8)\\(0, 2), (8, 2), (0, 6), (8, 6)} &8  \\
\hline   3 & \tabincell{c}{(4, 0), (4, 8), (0, 4), (8, 4)} & 16  \\
\hline   3 & \tabincell{c}{(1, 1), (1, 7), (7, 1), (7, 7)} & 16  \\
\hline   3 & \tabincell{c}{(3, 1), (3, 7), (5, 1), (5, 7)\\(1, 3), (7, 3), (1, 5), (7, 5)} &    112 \\
\hline   3 & \tabincell{c}{(2, 2), (2, 6), (6, 2), (6, 6)} & 192 \\
\hline   3 & \tabincell{c}{(4, 2), (4, 6), (2, 4), (6, 4)} & 496 \\
\hline   3 & \tabincell{c}{(3, 3), (3, 5), (5, 3), (5, 5)} & 784 \\
\hline   3 & \tabincell{c}{(4, 4)} & 1216 \\
\hline   ... & ... &          ... \\
\hline   9 & (1, 1), (7, 7) &   8 \\
\hline   9 & (3, 3), (5, 5) &   56 \\
\hline  10 & (2, 2), (6, 6) &   16 \\
\hline  10 & (4, 4) &           32 \\
\hline  11 & (2, 2), (6, 6) &   8  \\
\hline  11 & (4, 4) &           16 \\
\hline  12 & (4, 4) &           16 \\
\hline  13 & (2, 2), (6, 6) &   4  \\
\hline  14 & (4, 4) &           4 \\
\hline  15 & (4, 4) &           2 \\
\hline  16 & (8, 8) &           1 \\
\hline
\end{tabular}
\end{table}

\begin{remark}
Algorithm \ref{algorithm1} is an efficient method to calculate the split polar spectrum for the special case $L=2$. Theoretically, this algorithm can be extended to the general case $L>2$. However, in this case, the computational complexity of the split polar spectrum will become too high to apply for the performance evaluation. Hence, for a large block diversity, we use a random mapping as a tool to approximate the split polar spectrum which will be described in the next section.
\end{remark}

\section{Performance Analysis based on Random Mapping}
\label{section_V}
In this section, under the random mapping, we derive the upper bound of error probability of polarized channel and provide the approximate estimation of split polar spectrum for the general block diversity. Furthermore, we derive and analyze the approximation expression of upper bound.
\subsection{Upper Bound of Error Probability under Random Mapping}
Recall that the uniform interleaving $\Pi(\cdot)$ is used in the random mapping as shown in Fig. \ref{Figure1}. Given a codeword with weight $d$, the uniform interleaver deploys $d$ nonzero coded bits over the $L$ fading blocks. Let $f_v$ designate the number of fading blocks with weight $v$ and $w=\min(d,M)$ denote the maximum weight of one block. The codeword weight $d$ is related to the weight pattern $\mathbf{f}=(f_0,f_1,...,f_w)$, namely $d \to \mathbf{f}$. So we have $L=\sum_{v=0}^{w}{f_v}$ and $d=\sum_{v=1}^{w}{vf_v}$. Let $F=L-f_0$ denote the number of fading blocks with nonzero weights. Especially, let $\bm{\alpha}_v=\left(\alpha_{v,1},...,\alpha_{v,t},...,\alpha_{v,f_v}\right)$ denote the envelope subvector corresponding to $f_v$ fading blocks with the weight $v$. So the envelope vector is $\bm{\alpha}=\left(\bm{\alpha}_0,...,\bm{\alpha}_w\right)$. Similar to the derivation in \cite{RandomMapping}, we obtain the upper bound of error probability of the polarized channel as the following.
\begin{theorem}\label{theorem6}
Given the polar subcode ${\mathbb{D}}_N^{(i)}$, for the random mapping, the error probability of the polarized channel $W_N^{(i)}$ can be upper bounded as
\begin{equation}\label{equation34}
\begin{aligned}
P\left(W_N^{(i)}\right)&\leq \sum_{d=d_{min}^{(i)}}^{N} \sum_{F=\lceil{d/M}\rceil}^{d}\sum_{f_1=0}^{F_1}\sum_{f_2=0}^{F_2}...\sum_{f_w=0}^{F_w}  {A}_N^{(i)}(d)\\
&\cdot P_d(\mathbf{f})\prod\limits_{v=1}^{w} \left(\frac{1}{v {E_s}/{N_0}+1}\right)^{f_v},
\end{aligned}
\end{equation}
where $F_v=\min\left\{F-\sum_{r=1}^{v-1}{f_r},\frac{d-\sum_{r=1}^{v-1}{rf_r}}{v}\right\},1\leq v\leq w$ and $P_d(\mathbf{f})$ is the probability of a weight pattern $\mathbf{f}$ related to a specific weight $d$.
\end{theorem}
\begin{proof}
Due to the uniform interleaving, the error probability of the polarized channel $W_N^{(i)}$ is upper bounded by averaging over all possible fading block patterns and the envelopes, that is,
\begin{equation}\label{equation35}
P\left(W_N^{(i)}\right)\leq \sum_{d=d_{min}^{(i)}}^{N} \mathbb{E}_{\mathbf{f},\bm{\alpha}}\left[{A}_N^{(i)}(d) P_d(\mathbf{f})P(d|\bm{\alpha},\mathbf{f})\right].
\end{equation}

Here, given the weight pattern $\mathbf{f}$ and the envelope vector $\bm{\alpha}$, the conditional PEP can be upper bounded by
\begin{equation}\label{equation36}
  P\left(d\left|\bm{\alpha},\mathbf{f}\right.\right)\leq \prod_{v=1}^{w} \prod_{t=1}^{f_v} \exp\left( -\alpha_{v,t}^2 v\frac{E_s}{N_0} \right),
\end{equation}
where $\alpha_{v,t}$ denotes the envelope related to the $t$-th fading block with the same weight $v$.
The conditional PEP is further taken the expectation over the envelope $\bm{\alpha}$ and bounded by
\begin{equation}\label{equation37}
  \bar{P}\left(d\left| \mathbf{f}\right.\right)=\mathbb{E}_{\bm{\alpha}}\left[P\left(d\left| \bm{\alpha},\mathbf{f}\right.\right)\right]\leq \prod_{v=1}^{w} \left(\frac{1}{1+ v\frac{E_s}{N_0}}\right)^{f_v}.
\end{equation}
According to \cite{RandomMapping}, the probability of a fading block pattern for a specific codeword weight $d$ is written as
\begin{equation}\label{equation38}
\begin{aligned}
P_d(\mathbf{f})&=\frac{\prod\limits_{v=1}^{w}{\left(\begin{array}{c}M\\v\end{array}\right)}^{f_v}}{\left(\begin{array}{c}N\\d\end{array}\right)}\cdot\frac{L!}{\prod\limits_{v=0}^{w}{f_v!}}\\
&=\prod\limits_{v=1}^{w}{\left(\begin{array}{c}M\\v\end{array}\right)}^{f_v}B(L,\mathbf{f},d),
\end{aligned}
\end{equation}
where $B(L,\mathbf{f},d)={\left(\begin{array}{c}L\\f_0,f_1,...,f_w\end{array}\right)}\left/{\left(\begin{array}{c}N\\d\end{array}\right)}\right.$.
Note that the right term of the first equality in (\ref{equation38}) means the number of combinations of $\mathbf{f}$ among the $L$ blocks, which is a multinominal coefficient $\left(\begin{array}{c}L\\f_0,f_1,...,f_w\end{array}\right)$. Correspondingly, the left factor of this equality indicates the probability of deploying $d$ nonzero bits over $L$ blocks whereby the number of blocks with weight $v$ is $f_v$ for all possible values of $v$.
In addition, since the summation order in (\ref{equation34}) is from $f_1$ to $f_w$, we can derive the superscript of each summation as $F_v=\min\left\{F-\sum_{r=1}^{v-1}{f_r},\frac{d-\sum_{r=1}^{v-1}{rf_r}}{v}\right\}$. In fact, other summation order is also possible.
\end{proof}

Compared with the upper bound (\ref{equation12}) in Corollary \ref{corollary1}, we find that the upper bound in (\ref{equation34}) has a similar form. Recall that the calculation of split polar spectrum is very difficult when $L>2$. So we can regard $A_N^{(i)}P_d(\mathbf{f})$ as an approximation of split polar spectrum $A_N^{(i)}(\mathbf{d})$. Since the polar spectrum $A_N^{(i)}$ can be easily calculated by using the enumeration algorithm in \cite{PolarSpectrum_Niu}. With the help of uniform interleaving, this bound is suitable for arbitrary diversity order.

For the random mapping, we can also obtain the BLER upper bound as the following corollary.
\begin{corollary}\label{corollary3}
Given the fixed configuration $(N,K,\mathcal{A})$, for the random mapping, the block error probability of polar code using SC decoding in the block Rayleigh fading channel is bounded by
\begin{equation}\label{equation39}
P_e(N,K,\mathcal{A})\leq \! \sum_{i\in \mathcal{A}}\! \sum_{d \to \mathbf{f}, d=d_{min}^{(i)}}^N \!\!\!\!\!\!\!A_N^{(i)}(d) P_d(\mathbf{f})\!\prod_{v=1}^{w} \left(\frac{1}{1+vE_s/N_0}\right)^{f_v}.
\end{equation}
\end{corollary}

\subsection{Approximation Analysis of Upper Bound}
Given the weight pattern $\mathbf{f}$, for the case of high SNR, the conditional PEP can be further bounded by
\begin{equation}\label{equation40}
  \bar{P}\left(d\left| \mathbf{f}\right.\right)\leq \prod_{v=1}^{w} \left(\frac{1}{v\frac{E_s}{N_0}}\right)^{f_v}
   =\left(\prod\limits_{v=1}^{w}{v^{f_v}}\right)^{-1} \left(\frac{E_s}{N_0}\right)^{-\left(L-f_0\right)}.
\end{equation}

Like the analysis in Section \ref{DesignCriteria}, we find that a diversity order in term of $L-f_0$ can be achieved in Eq. (\ref{equation40}). If the number of zero-weight blocks is zero, i.e., $f_0=0$, we will obtain the full diversity. Similarly, by Lemma \ref{lemma1}, as a product distance, the term $\prod_{v=1}^{w}{v^{f_v}}$ can be approximately maximized when $w=d/L$ and $f_w=L$. Therefore, the design criteria of polar coding in Section \ref{DesignCriteria} are also suitable for the polar code struction under the random mapping.
\begin{remark}
Guided by the full diversity criterion, we require the interleaver in random mapping can uniformly distribute the nonzero bits over $L$ fading blocks so as to achieve the $L$ diversity order. Hence, if the uniform interleaver is carefully designed and satisfies the equipartition condition, the full diversity can be achieved and we have $f_0=0$. Furthermore, since the low-weight terms dominate the upper bound of the error probability in (\ref{equation34}), guided by the product distance criterion, for the random mapping, we can also optimize the product distance corresponding to the minimum Hamming distance $d_{min}^{(i)}$ and obtain a coding advantage. Certainly, the influence of polar spectrum $A_N^{(i)}$ and the probability of fading block pattern $P_d(\mathbf{f})$ should be elaborately considered.
\end{remark}

\section{Polar Coding Based on Random Mapping}
\label{section_VI}
In this section, we consider the construction method of polar codes under the random mapping. Based on the logarithmic version of the upper bound, the constructive metric is also named polarized diversity weight (PDW).

For the random mapping, we can also consider the upper bound of error probability in the case of low-SNR and obtain the following theorem.
\begin{theorem}\label{theorem8}
In the case of low-SNR and random mapping, the upper bound on the error probability of the polarized channel can be approximated by
\begin{equation}
P\left(W_N^{(i)}\right)\lesssim \sum_{d=d_{min}^{(i)}}^{N} \sum_{\mathbf{f}} \prod_{v=1}^{w} { \frac{e^{\left[-f_v\left(1+v \frac{E_s}{N_0}-\ln C_N^{(i)}(f_v)\right)\right]}}{\left[1+v \frac{E_s}{N_0}-\ln C_N^{(i)}(f_v)\right]^{f_v}} },
\end{equation}
where $C_N^{(i)}(f_v)=\left[A_N^{(i)}(d)B(L,\mathbf{f},d)\right]^{\frac{1}{wf_v}}\left(\begin{array}{c}M\\v\end{array}\right)$.
\end{theorem}
\begin{proof}
Let $\gamma_{v,t}=\alpha_{v,t}\frac{E_s}{N_0}$ denote the fading block SNR corresponding to the $t$-th block with the weight $v$. After decomposing the term $A_N^{(i)}P_d(\mathbf{f})$ into the inner of the product, we can rewrite the upper bound of the error probability of polarized channel as
\begin{equation}\label{equation44}
\begin{aligned}
&P\left(W_N^{(i)}\right)\leq \sum_{d} \sum_{\mathbf{f}} \prod_{v=1}^{w} \prod_{t=1}^{f_v} C_N^{(i)}(f_v) \int\nolimits_0^{\infty} e^{-\gamma_{v,t} v }g(\gamma_{v,t})d\gamma_{v,t}\\
& =\sum_{d} \sum_{\mathbf{f}} \prod_{v=1}^{w} \prod_{t=1}^{f_v}\int\nolimits_0^{\infty} e^{\ln C_N^{(i)}(f_{v,t})-\gamma_{v,t} v }\frac{1}{\gamma}e^{-\gamma_{v,t}/\gamma}d\gamma_{v,t}.
\end{aligned}
\end{equation}
The inner integration in (\ref{equation44}) can also be decomposed into two parts as follows,
\begin{equation}
\begin{aligned}
&\int\nolimits_0^{\infty} \frac{1}{\gamma}e^{\ln C_N^{(i)}(f_v)-\gamma_{v,t} \left(v+\frac{1}{\gamma}\right) }d\gamma_{v,t}\\
&=\!\int\nolimits_0^{\gamma} \frac{1}{\gamma}e^{\ln C_N^{(i)}(f_v)-\gamma_{v,t} \left(v+\frac{1}{\gamma}\right) }d\gamma_{v,t}\!\\
&+\!\int\nolimits_{\gamma}^{\infty} \frac{1}{\gamma}e^{\ln C_N^{(i)}(f_v)-\gamma_{v,t} \left(v+\frac{1}{\gamma}\right) }d\gamma_{v,t}\!=I_1+I_2.
\end{aligned}
\end{equation}
For the case of low-SNR, i.e., $\gamma\approx 0$, the first integration $I_1$ tends to zero. Thus, the second integration $I_2$ is concerned and further bounded by
\begin{equation}\label{equation46}
\begin{aligned}
I_2 &\leq \int\nolimits_{\gamma}^{\infty} \frac{1}{\gamma}e^{\ln C_N^{(i)}(f_v)\frac{\gamma_{v,t}}{\gamma}-\gamma_{v,t} \left(v+\frac{1}{\gamma}\right) }d\gamma_{v,t}\\
&=\frac{\exp{\left(-\left(1+\gamma v-{\ln C_N^{(i)}(f_v)}\right)\right)}}{1+\gamma v-{\ln C_N^{(i)}}(f_v)}.
\end{aligned}
\end{equation}
In the derivation of (\ref{equation46}), we use the inequality $\frac{\gamma_{v,t}}{\gamma}\geq 1$.
\end{proof}

Similar to Theorem \ref{theorem3}, the SNR in Theorem \ref{theorem8} should satisfy the condition $ \forall v, v E_s/N_0+1-{\ln C_N^{(i)}}(f_v)>0$. In the case of low SNR, by using Theorem \ref{theorem8}, the logarithmic version of the upper bound can be written as
\begin{equation}
\begin{aligned}
&\ln \left\{\sum_{d=d_{min}^{(i)}}^{N} \sum_{\mathbf{f}} \prod_{v=1}^{w} {   \frac{\exp\left[-f_v\left(1+v \frac{E_s}{N_0}-\ln C_N^{(i)}(f_v)\right)\right]}{\left[1+v \frac{E_s}{N_0}-\ln C_N^{(i)}(f_v)\right]^{f_v}}  }\right\}\\
&\begin{aligned}\propto \max_{\mathbf{f}}
     &\left\{ \sum_{v=1}^{w}f_v\ln {C_N^{(i)}(f_v)}-d_{min}^{(i)}\frac{E_s}{N_0} \right.\\
     &\left.-\sum_{v=1}^{w}f_v\ln \left[{1+v \frac{E_s}{N_0}-\ln C_N^{(i)}(f_v)}\right]\right\}.
    \end{aligned}
\end{aligned}
\end{equation}
Similarly, the approximation $\ln\left( \sum\limits_{k} e^{a_k} \right)\approx \max\limits_{k}\{a_k\}$ is used. Given the assumption of full diversity, two equalities $\sum_{v=1}^{w}v f_v=d$ and $\sum_{v=1}^{w} f_v=L$ are used to simplify the metric.
When $v \frac{E_s}{N_0}-\ln C_N^{(i)}(f_v)$ is sufficiently small, by using the approximation $\ln(1+x)\approx x$, we obtain the PDW as below.

\begin{metric}\label{theorem9}
Given the code length $N$, for the random mapping, the polarized diversity weight (PDW) can be given by,
\begin{equation}\label{equation48}
\begin{aligned}
PDW_N^{(i)}=\max_{\mathbf{f},\sum\limits_{v=1}^{w}{vf_v}=d_{min}^{(i)}} &\left[L_N^{(i)}\left(d_{min}^{(i)}\right)+\ln B\left(L,\mathbf{f},d_{min}^{(i)}\right)\right.\\
&\left.+\sum\limits_{v=1}^{w} f_v\ln{\left(\begin{array}{c}M\\v\end{array}\right)}-d_{min}^{(i)} \frac{E_s}{N_0}\right].
\end{aligned}
\end{equation}
\end{metric}

\begin{remark}
Compared with the construction under the block mapping, the polar codes constructed under the random mapping may achieve the same diversity gain since the carefully designed interleaver can uniformly distribute the nonzero bits over all the fading blocks. In addition, only simple polar spectrum rather than complex split polar spectrum is needed for the calculation of the construction metrics under the random mapping. Therefore, such PDW metric is attractive for the practical implementation of polar coding.
\end{remark}

\section{Numerical Analysis and Simulation Results}
\label{section_VII}
In this section, we will provide the numerical and simulation results of polar codes based on the split polar spectrum under the block Rayleigh fading channel with $L$ fading blocks. First, we present the BLER simulation results in the case of $L\leq 2$, where both the block mapping and random mapping are considered. Then, the BLER upper bounds and simulation results in the case of $L>2$ are analyzed and compared, where only the random mapping is adopted.
\subsection{Results under $L\leq2$}
In this part, under the condition of $L\leq2$, we compare the BLER simulation performances of polar codes generated by the proposed PDW construction and the traditional methods, which include GA algorithm and the one based on Bhattacharyya parameter according to \cite{BlockFading_Bravo-Santos}. Note that the construction metric PDW under the block mapping is equivalent to that under the random mapping when $L=1$. Otherwise, the constructions under the two mappings are marked by ``block'' and ``random'', respectively. While the GA algorithm is executed on each symbol SNR.

\begin{figure*}[htbp]
\setlength{\abovecaptionskip}{0.cm}
\setlength{\belowcaptionskip}{-0.cm}
  \centering{\includegraphics[scale=0.9]{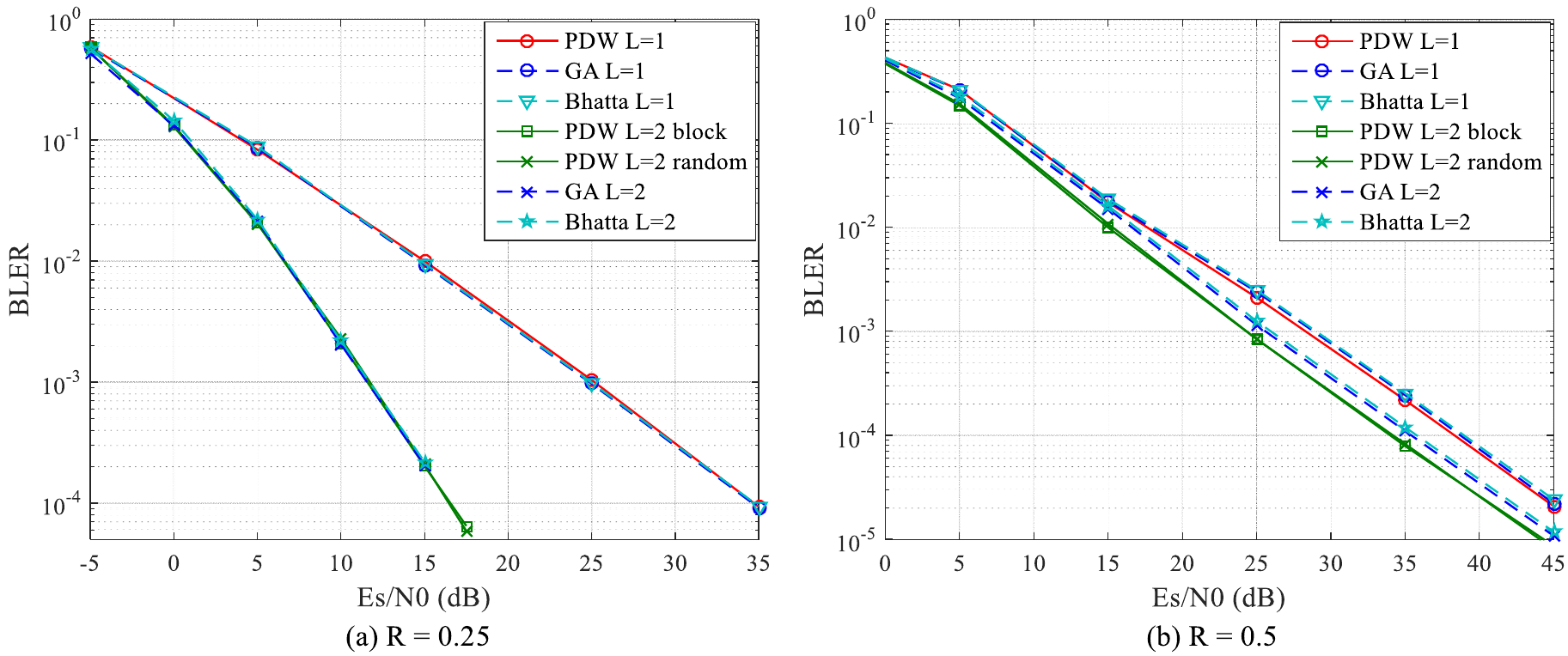}}
  \caption{The BLER performances comparison between the polar codes constructed based on GA, Bhattacharyya parameter and PDW under block Rayleigh fading channel, where $N = 256$, $R=\{0.25, 0.5\}$ and SC decoding is used.}\label{N256_K64_128}
\end{figure*}

Given the code length $N=256$, Fig. \ref{N256_K64_128} provides the BLER performances comparison among the various constructions under SC decoding. For the PDW construction, the symbol SNR in (\ref{equation24})/(\ref{equation48}) is respectively set to $0$ and $3$ dB for the code rate $R=0.25$ and $R=0.5$. It can be observed from Fig. \ref{N256_K64_128}(a) that the polar codes constructed by PDW can achieve the same performance of those constructed by Bhattacharyya parameter or GA algorithm which involves complex iterative calculation when the low code rate $R=0.25$ is considered. Moreover, in the case of $L=2$, the PDW under the random mapping behaves almost identically as that under the block mapping, which implies the approximate split polar spectrum derived by the random mapping is an effective alternative for the exact one under the block mapping. Meanwhile, the performances under $L=2$ exhibit obvious diversity gain comparing to that under $L=1$, which is consistent with the theoretical analysis of diversity order. Similar observations can be obtained in the Fig. \ref{N256_K64_128}(b) for the medium code rate $R=0.5$. However, the diversity advantage of $L=2$ against $L=1$ shrinks due to the fact that the diversity order depends both on the number of fading blocks $L$ as well as on the code rate $R$ \cite{CodedDiversity}.

\subsection{Results under $L>2$}
In this part, under $L=4$, we compare the BLER upper bounds and simulation results of polar codes generated by the proposed PDW construction and the traditional methods. For the PDW, only the random mapping is adopted to avoid the complex calculation of exact split polar spectrum. Meanwhile, the BLERs under $L\leq2$ are also provided for comparison.

Fig. \ref{N256_K64_upper_bound} provides the upper bound of BLER as well as the BLER performances of polar codes constructed by PDW with $N=256$ and $R=0.25$ under SC decoding. The upper bounds of BLER marked by the dash line are obtained by (\ref{equation13}). As shown in Fig. \ref{N256_K64_upper_bound}, all the upper bounds of BLER dramatically decrease with the increase of symbol SNR, and the slopes of curves are identical with those of the corresponding simulation results for PDW construction. Although the proposed upper bound is quite loose, it can be calculated with a linear complexity and used to deduce construction metrics with explicit expressions.

\begin{figure}[htbp]
\setlength{\abovecaptionskip}{0.cm}
\setlength{\belowcaptionskip}{-0.cm}
  \centering{\includegraphics[scale=0.6]{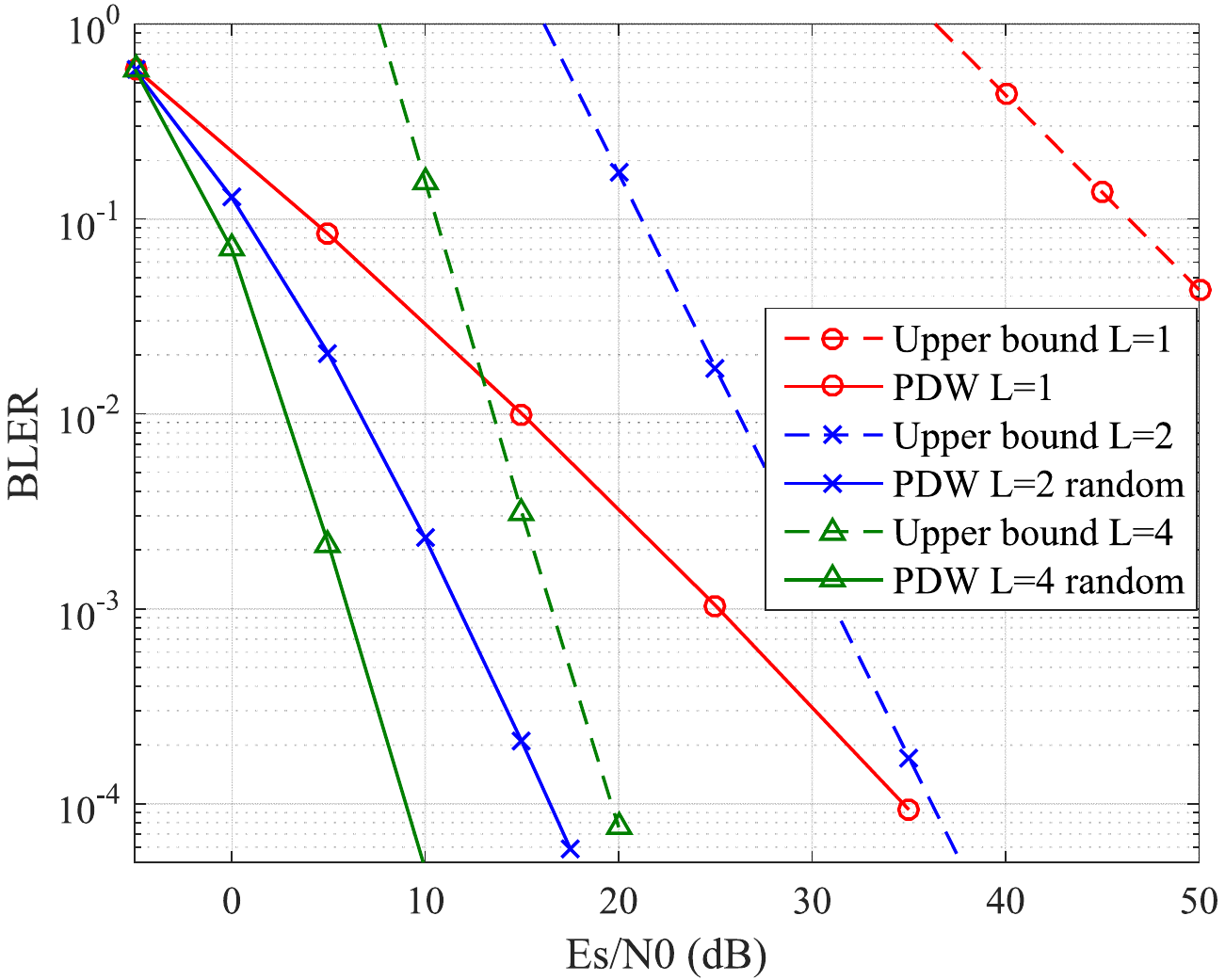}}
  \caption{The BLER performance and upper bound of the polar codes constructed based on PDW under block Rayleigh fading channel, where $N = 256$, $R=0.25$ and SC decoding is used.}\label{N256_K64_upper_bound}
\end{figure}

For the code length $N=1024$ and code rate $R=0.25$, the BLER performance under SC decoding is shown in Fig. \ref{N1024_K256_L=4}(a). The symbol SNR in (\ref{equation48}) is set to 0 dB for the PDW construction. It can be observed that the polar codes constructed by PDW perform superior to those constructed by GA algorithm and Bhattacharyya parameter, especially under the condition of $L=2$ and $L=4$. Specifically, for $L=2$, comparing to the GA method or the one based on Bhattacharyya parameter, both diversity gain and coding gain can be obtained by PDW construction, while for $L=4$, the PDW construction exhibits notable coding gain against the two methods. Actually, the PDW construction benefits from its design criteria, i.e, the full diversity criterion and product distance criterion. Furthermore, the BLER performance for PDW construction under SCL decoding with list size 16 is depicted in Fig. \ref{N1024_K256_L=4}(b), which indicates that the BLER performance under the block fading channels is dominated by the number of fading blocks $L$ and only relatively marginal improvement can be achieved by employing the SCL decoding.

\begin{figure*}[htbp]
\setlength{\abovecaptionskip}{0.cm}
\setlength{\belowcaptionskip}{-0.cm}
  \centering{\includegraphics[scale=0.9]{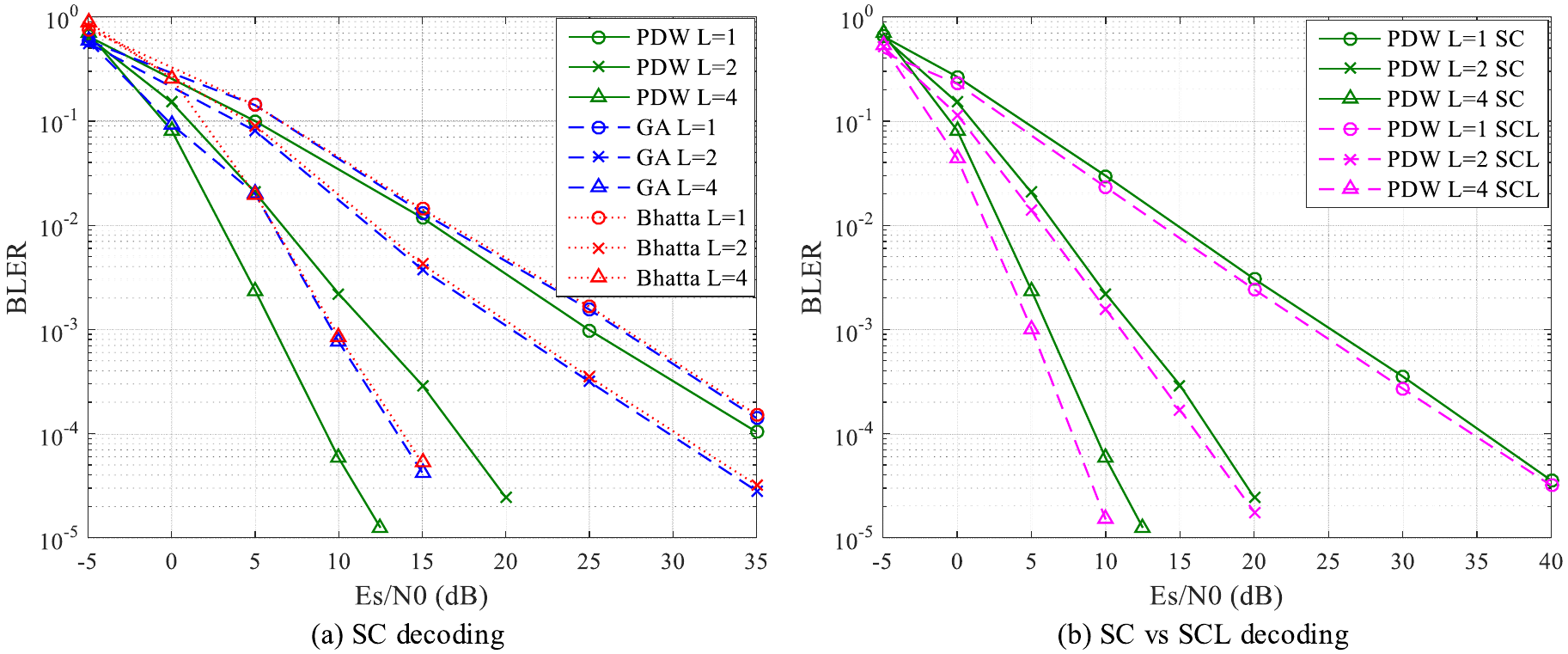}}
  \caption{The BLER performances comparison among the polar codes constructed based on GA, Bhattacharyya parameter and PDW under block Rayleigh fading channel, where $N = 1024$, $R=0.25$ and SC decoding or SCL decoding with list size 16 is used.}\label{N1024_K256_L=4}
\end{figure*}

Finally, in Fig. \ref{outage}, the BLER performances of polar codes constructed by PDW are compared with the capacity outage under different code rate $R$, where the capacity outages are similarly computed numerically by using Eq. (30) in \cite{On_coding} and are marked by ``$P_{out}$''. For the two code rate $R=0.25$ and $0.6$, the symbol SNRs in (\ref{equation48}) are respectively set to $0$ and $3$ dB. As shown in Fig. \ref{outage}, the proposed PDW construction can achieve the same diversity gain as the capacity outage in both case of $L=2$ and $L=4$. Therefore, the proposed PDW is an attractive construction metric for practical application which can guarantee maximum diversity with a linear complexity.

\begin{figure}[htbp]
\setlength{\abovecaptionskip}{0.cm}
\setlength{\belowcaptionskip}{-0.cm}
  \centering{\includegraphics[scale=0.65]{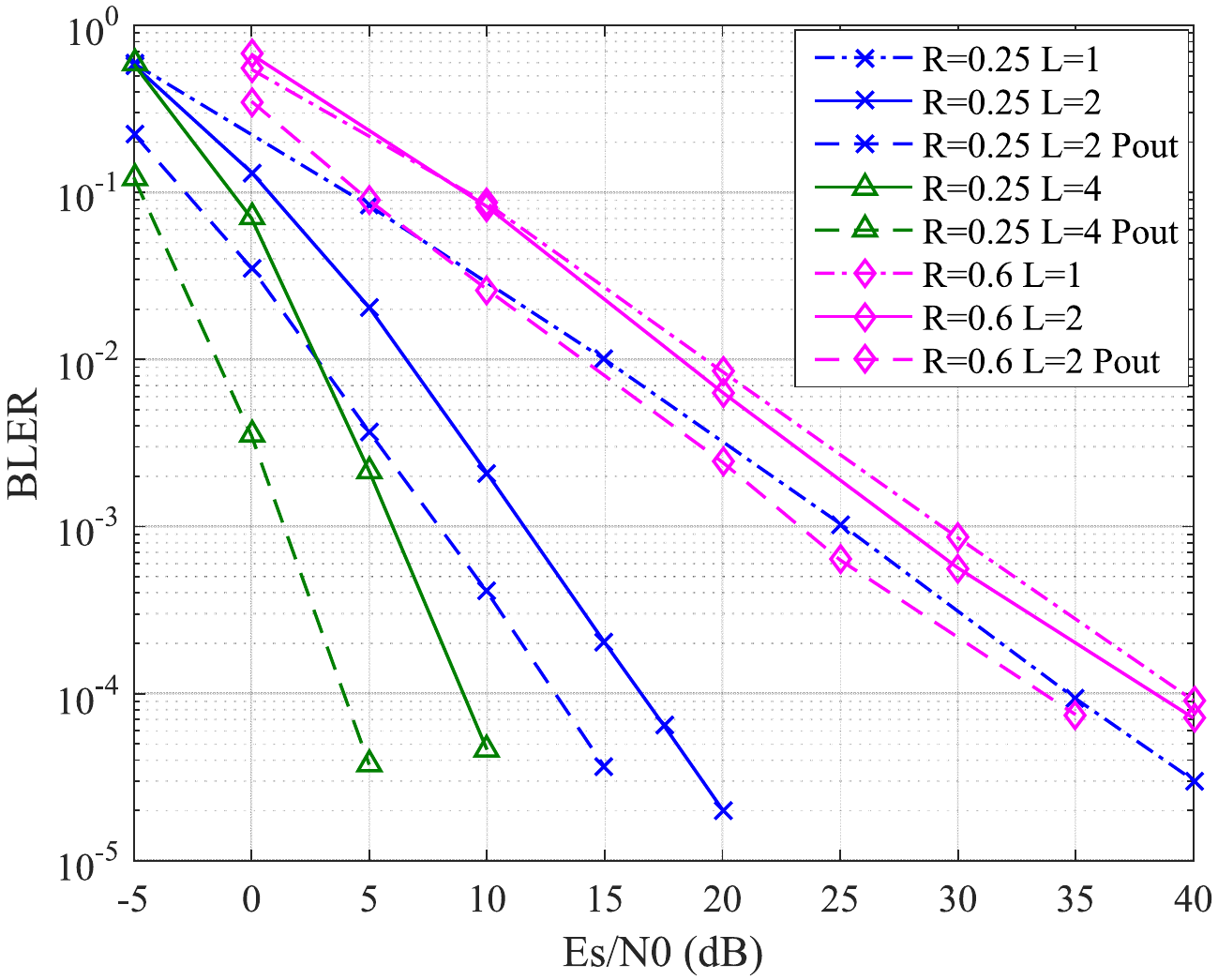}}
  \caption{The comparison between the BLER performances and capacity outage under block Rayleigh fading channel, where $N = 256$ and SC decoding is used for PDW construction.}\label{outage}
\end{figure}

\section{Conclusions}
\label{section_VIII}
In this paper, we establish a systematic framework to analyze the theoretical performances of polar codes under the block Rayleigh fading channels by introducing a new concept named split polar spectrum. In the case of block mapping and random mapping, based on the split polar spectrum, we derive the upper bound on the error probability of the polarized channel as well as that on the BLER of SC decoding. In addition, we propose an enumeration algorithm embedded the solution of general MacWilliams identities to calculate the exact split polar spectrum for the special case $L=2$ under the block mapping. While for arbitrary diversity order, the approximate split polar spectrum can be derived under the random mapping by means of uniform interleaving. Finally, we propose two design criteria, i.e., the full diversity criterion and the product distance criterion, to construct polar codes over the block Rayleigh fading channels, whereby an explicit and analytical construction metric named PDW is proposed to construct polar codes with a linear complexity. Simulation results show that the polar codes constructed by the PDW metric can achieve similar or superior performance to those constructed by GA algorithm and Bhattacharyya parameter under SC decoding.



\end{document}